\documentclass[runningheads]{llncs}
\usepackage{amsmath}
\usepackage{tikz}
\usetikzlibrary{shapes,decorations,arrows,shadows,mindmap,trees}
\usepackage{hyperref}
\hypersetup{
	colorlinks=true,
	linkcolor=blue,
	filecolor=magenta,      
	urlcolor=cyan,
}
\newcommand{\A}{\mathcal{A}}
\newcommand{\B}{\mathcal{B}}
\newcommand{\C}{\mathcal{C}}
\newcommand{\U}{\mathcal{U}}
\newcommand{\Ub}[1]{\mathcal{U}_\mathcal{B}{#1}}
\newcommand{\Uc}[1]{\mathcal{U}_\mathcal{C}{#1}}
\newcommand{\UB}[1]{\mathcal{U}^2_\mathcal{B}{#1}}
\newcommand{\UC}[1]{\mathcal{U}^2_\mathcal{C}{#1}}
\renewcommand{\S}{\mathcal{S}}
\renewcommand{\O}{\mathcal{O}}
\spnewtheorem{observation}[theorem]{Observation}{\bfseries}{\itshape}
\spnewtheorem{result}{Result}{\bfseries}{\itshape}
\usepackage{graphicx}
\usepackage{array}
\usepackage{float}
\usepackage{wrapfig}
\usepackage{amssymb,amsmath}
\usepackage[style=base]{caption}
\usepackage{wrapfig}
\begin{document}

\title{Improved Bounds for Two Query Adaptive Bitprobe Schemes Storing Five Elements}

\titlerunning{Two Query Five Element Bitprobe Schemes}

\author{Mirza Galib Anwarul Husain Baig \and
Deepanjan Kesh}

%

\institute{Indian Institute of Technology Guwahati, Guwahati, Assam 781039, India \email{\{mirza.baig,deepkesh\}@iitg.ac.in}}

\maketitle

\begin{abstract}
	In this paper, we study two-bitprobe adaptive schemes storing five
	elements. For these class of schemes, the best known lower bound is
	$\Omega(m^{1/2})$ due to Alon and Feige~\cite{DBLP:conf/soda/AlonF09}.
	Recently, it was proved by Kesh~\cite{DBLP:conf/fsttcs/Kesh18} that
	two-bitprobe adaptive schemes storing three elements will take at least
	$\Omega(m^{2/3})$ space, which also puts a lower bound on schemes
	storing five elements. In this work, we have improved the lower bound to
	$\Omega(m^{3/4})$. We also present a scheme for the same that takes
	$\O(m^{5/6})$ space. This improves upon the $\O(m^{18/19})$-scheme due
	to Garg~\cite{garg15} and the $\O(m^{10/11})$-scheme due to Baig {\em et
	al.}~\cite{DBLP:conf/walcom/BaigKS19}.

	\keywords{Data structure \and Set membership problem \and Bitprobe model
	\and Adaptive Scheme.}
\end{abstract}

\section{Introduction}

	The {\em static membership problem} involves the study and construction of such data structures which can store an arbitrary subset $\S$ of size at most $n$ from the universe $\U = \{ 1, 2, 3, \dots, m \}$ such that membership queries of the form ``Is $x$ in $\S$?'' can be answered correctly and efficiently. A special category of the static membership problem is the {\em bitprobe model} in which we evaluate our solutions w.r.t. the following resources -- the size of the data structure, $s$, required to store the subset $\S$, and the number of bits, $t$, of the data structure read to answer membership queries. It is the second of these resources that lends the name to this model.
	
	In this model, the design of data structures and query algorithms are known as {\em schemes}. For a given universe $\U$ and a subset $\S$, the algorithm to set the bits of our data structure to store the subset is called the {\em storage scheme}, whereas the algorithm to answer membership queries is called the {\em query scheme}. Schemes are divided into two categories depending on the nature of our query scheme. Upon a membership query for an element, if the decision to probe a particular bit depends upon the answers received in the previous bitprobes of this query, then such schemes are known as {\em adaptive schemes}. If the locations of the bitprobes are fixed for a given element of $\U$, then such schemes are called {\em non-adaptive schemes}.
	
	For any particular setting of $n, m, s,$ and $t$, the corresponding scheme is referred to in the literature as a $(n, m, s, t)$-scheme~\cite{DBLP:conf/stoc/BuhrmanMRV00,DBLP:conf/esa/RadhakrishnanRR01,DBLP:conf/esa/RadhakrishnanSS10}. Radhakrishnan {\em et al.}~\cite{DBLP:conf/esa/RadhakrishnanSS10} also introduced the convenient notations $s_A(n, m, t)$ and $s_N(n, m, t)$ to denote the space required by an adaptive or a non-adaptive scheme, respectively.

\subsection{Two-bitprobe Adaptive Schemes}

	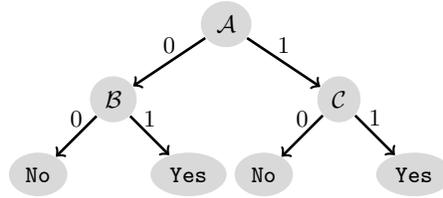
\begin{figure}[t]
		\centering
		\begin{tikzpicture}[scale=0.5]
			\node[ellipse, fill=gray!30] (A) at (5,4) {$\mathcal{A}$};
			\node[ellipse, fill=gray!30] (B) at (2,2) {$\mathcal{B}$};
			\node[ellipse, fill=gray!30] (C) at (8,2) {$\mathcal{C}$};
			\node[ellipse, fill=gray!30] (D) at (0,0) {{\tt No}};
			\node[ellipse, fill=gray!30] (E) at (4,0) {{\tt Yes}};
			\node[ellipse, fill=gray!30] (F) at (6,0) {{\tt No}};
			\node[ellipse, fill=gray!30] (G) at (10,0) {{\tt Yes}};

			\draw[line width=1pt, ->] (A) -- (B) node[midway, above] {$0$};
			\draw[line width=1pt, ->] (A) -- (C) node[midway, above] {$1$};

			\draw[line width=1pt, ->] (B) -- (D) node[midway, above] {$0$};
			\draw[line width=1pt, ->] (B) -- (E) node[midway, above] {$1$};
	
			\draw[line width=1pt, ->] (C) -- (F) node[midway, above] {$0$};
			\draw[line width=1pt, ->] (C) -- (G) node[midway, above] {$1$};
		\end{tikzpicture}
		\caption{The decision tree of an element.}
		\label{fig:tree}
	\end{figure}

	In this paper, we consider those schemes that use two bitprobes ($t = 2$) to answer membership queries. Data structures in such schemes consist of three tables, namely $\A, \B,$ and $\C$. The first bitprobe is always made in table $\A$; the location of the bit being probed, of course, depends on the element which is being queried. The second bitprobe is made in table $\B$ or in table $\C$, depending on whether $0$ was returned in the first bitprobe or $1$ was returned. The final answer of the query scheme is {\tt Yes} if $1$ is returned by the second bitprobe, otherwise it is {\tt No}. The data structure and the query scheme can be succintly denoted diagrammatically by what is known as the {\em decision tree} of the scheme (Figure~\ref{fig:tree}).

\subsection{Our Contribution}

	In this paper, we study schemes when the number of allowed
	bitprobes is two ($t = 2$) and the subset size is at most five ($n = 5$).
	Some progress has been made for subsets of smaller sizes.

	When the subset size is odd, the problem is well understood. More
	particularly, when $n = 1$, there exists a scheme that takes
	$\O(m^{1/2})$ amount of space, and it has been shown that it matches
	with the lower bound
	$\Omega(m^{1/2})$~\cite{DBLP:conf/soda/AlonF09,DBLP:conf/stoc/BuhrmanMRV00,DBLP:conf/esa/LewensteinMNR14}.
	When $n = 3$, Baig and Kesh~\cite{DBLP:conf/walcom/BaigK18} have shown
	that there exists a $\O(m^{2/3})$-scheme, and
	Kesh~\cite{DBLP:conf/fsttcs/Kesh18} has proven that it matches with the lower bound
	$\Omega(m^{2/3})$.

	For even sized subsets, tight bounds are yet to be proven. For $n = 2$,
	Radhakrishnan {\em et al.}~\cite{DBLP:conf/esa/RadhakrishnanRR01} have
	proposed a scheme that takes $\O(m^{2/3})$ space, and for $n = 4$, Baig
	{\em et al.}~\cite{DBLP:conf/iwoca/BaigKS19} have presented a
	$\O(m^{5/6})$-scheme, but it is as of yet unknown whether these bounds
	are tight.
	
	For subsets of size five ($n = 5$), the best known lower bound was due to Alon and Feige~\cite{DBLP:conf/soda/AlonF09} which is $\Omega(m^{1/2})$. The $\Omega(m^{2/3})$ lower bound for $n = 3$ also puts an improved bound for the $n = 5$ case. Our first result improves the bound to $\Omega(m^{3/4})$.
	
	\begin{result}
		[Theorem~\ref{thm:lower}] $s_A(5, m, 2) = \Omega(m^{3/4})$.
	\end{result}
	
	We also propose an improved scheme for the problem. The best known upper bound was due to Garg~\cite{garg15} which was $\O(m^{18/19})$, which was improved by Baig {\em et al.}~\cite{DBLP:conf/walcom/BaigKS19} to $\O(m^{10/11})$. In this paper, we improve the bound to $\O(m^{5/6})$.
	
	\begin{result}
		[Theorem~\ref{thm:upper}] $s_A(5, m, 2) = \O(m^{5/6})$.
	\end{result}
	
	One thing to note is that the space for the scheme storing five elements now matches the space for the scheme storing four elements. Moreover, the two results stated above combined together significantly reduces the gap between the upper and lower bounds for the problem under consideration.

\section{Lower Bound}

	In this section, we present our proof for the lower bound of $s_A(5, m, 2)$. So, the size of our subset $\S$ that we want to store in our data structure is at most five.
	
	In table $\A$ of our data structure, multiple elements must necessarily map to the same bit to keep the table size to $o(m)$. The set of elements that map to the same bit in this table is referred to in the literature as a {\em block} (Radhakrishnan {\em et al.}~\cite{DBLP:conf/esa/RadhakrishnanRR01}). We refer by $\A(e)$ to the block to which the element $e$ belongs. Elements mapping to the same bit in tables $\B$ and $\C$ will be referred to as just {\em sets}. That set of table $\B$ to which the element $e$ belongs will be denoted by $\B(e)$. $\C(e)$ is similarly defined.
	
	Storing a member $e$ of our subset $\S$ in table $\B$ is an informal way to state the following -- the bit corresponding to $\A(e)$ is set to $0$, and $\B(e)$ is set to $1$. So, upon query for the element $e$, we will get a $0$ in our first bitprobe, query table $\B$ at location $\B(e)$ to get $1$, and finally answer {\tt Yes}. Similarly, storing an elememt $f$ which is not in $\S$ in table $\C$ would entail assigning $1$ to $\A(f)$ and $0$ to $\C(f)$.

	To start with, we make the following simplifying assumptions about any scheme for the aforementioned problem.
	\begin{enumerate}
		\item All the tables of our datastructure have the same size, namely $s$, and hence the size our data structure is $3 \times s$.
		\item If two elements belong to the same block in table $\A$, they do not belong to the same sets in either of tables $\B$ or $\C$.
	\end{enumerate}
	In the conclusion of this section, we will show that these assumptions do not affect the space asymptotically, but rather by constant factors.

\subsection{Universe of an Element}

	We now define the notion of the {\em universe} of an element. This is similar to the definition of the universe of a set in Kesh~\cite{DBLP:conf/fsttcs/Kesh18}.
	
	\begin{definition}
		\label{def:u}
		The universe of an element $e$ w.r.t. to table $\B$, denoted by $\Ub(e)$, is defined as follows.
		\[
			\Ub(e) = \bigcup_{ f \in \B(e) \setminus \{ e \}} \A(f) \setminus \{ f \}.
		\]
		Similarly, the universe of an element $e$ w.r.t. to table $\C$, denoted by $\Uc(e)$, is defined as follows.
		\[
			\Uc(e) = \bigcup_{ f \in \C(e) \setminus \{ e \}} \A(f) \setminus \{ f \}.
		\]
	\end{definition}
	
	A simple property of the universe of an element, which will be useful later, is the following.
	
	\begin{observation}
		\label{obs:e}
		\begin{enumerate}
			\item $\A(e) \cap \Ub(e) = \phi$ and $\B(e) \cap \Ub(e) = \phi$.
			\item $\A(e) \cap \Uc(e) = \phi$ and $\C(e) \cap \Uc(e) = \phi$.
		\end{enumerate}
	\end{observation}
	
	\begin{proof}
		Due to Assumption~2, $e$ is the only element of the block $\A(e)$ in set $\B(e)$. So, it follows from the definition that no element of $\A(e)$ is part of $\Ub(e)$. No element of $\B(e)$ is part of $\Ub(e)$ as we specifically preclude those elements from $\Ub(e)$. The other scenarios can be similarly argued. \qed
	\end{proof}
	
	We make the following simple observations about the universe of an element to help illustrate the constraints any storage scheme must satisfy to correctly store a subset $\S$.
	
	\begin{observation}
		\label{obs:ub}
		If $\B(e) \cap \S = \{ e \}$, and we want to store $e$ in table $\B$, then all the elements of $\Ub(e)$ must be stored in table $\C$.
	\end{observation}
	
	\begin{proof}
		As $e$ has been stored in table $\B$, the bit corresponding to the set $\B(e)$ must be set to $1$. Consider an element $f$, different from $e$, in $\B(e)$. According to Assumption~2, $e$ and $f$ cannot belong to the same block. As $f \notin \S$, so $f$ cannot be stored in table $\B$; if we do so, the query for $f$ will look at the bit $\B(e)$ and incorrectly return $1$. So, the element $f$ and its block $\A(f)$ must be stored in table $\C$.
		
		The above argument applies to any arbitrary element of $\B(e) \setminus \{ e \}$. So, according to Definition~\ref{def:u}, all of the elements of $\Ub(e)$ must be stored in table $\C$. \qed
	\end{proof}
	
	We make the same observation, without proof, in the context of table $\C$.

	\begin{observation}
		\label{obs:uc}
		If $\C(e) \cap \S = \{ e \}$, and we want to store $e$ in table $\C$, then all the elements of $\Uc(e)$ must be stored in table $\B$.
	\end{observation}
	
	Next, we define, what could be referred to as, a higher-order universe of an element, built on top of the universe of the element.
	
	\begin{definition}
		\label{def:u2}
		The 2-universe of an element $e$ w.r.t. table $\B$, denoted by $\UB(e)$, is defined as follows.
		\[
			\UB(e) = \bigcup_{f \in \Ub(e)} \C(f) \setminus \{ f \}.
		\]
		Similarly, the 2-universe of an element $e$ w.r.t. table $\C$, denoted by $\UC(e)$, is defined as follows.
		\[
			\UC(e) = \bigcup_{f \in \Uc(e)} \B(f) \setminus \{ f \}.
		\]
	\end{definition}

	The following observations provide more constraints for our storage schemes. 
	
	\begin{observation}
		\label{obs:u2b1}
		Consider an element $e$ such that $\B(e) \cap \S = \{ e \}$, and suppose we want to store $e$ in table $\B$. If $f$ is a member of $\Ub(e)$ such that $\C(f) \cap \S = \{ f \}$, then all the other members of $\C(f)$ must be stored in table $\B$.
	\end{observation}
	
	\begin{proof}
		If $e$, a member of $\S$, is stored in table $\B$, then Observation~\ref{obs:ub} tells us that all members of $\Ub(e)$ must be stored in table $\C$. As $f \in \Ub(e) \cap \S$, so $f$ must be stored in table $\C$, and consequently, the bit corresponding to $\C(f)$ must be set to $1$. As the other members of $\C(f)$ do not belong to $\S$, they cannot be stored in table $\C$, and hence must be stored in table $\B$. \qed
	\end{proof}

	\begin{observation}
		\label{obs:u2b2}
		Consider an element $e$ such that $\B(e) \cap \S = \{ e \}$, and suppose we want to store $e$ in table $\B$. If $f$ is a member $\Ub(e)$ such that $\C(f) \cap \S = \{ x \}$, where $x \neq f$, then $x$ must be stored in table $\B$.
	\end{observation}
	
	\begin{proof}
		As $e$, a member of $\S$, is stored in $\B$, Observation~\ref{obs:ub} tells us that all the members of $\Ub(e)$, and $f$ in particular, must be stored in table $\C$. As $f \notin \S$, so the bit corresponding to $\C(f)$ has to be $0$. As $x \in \C(f)$ belongs to $\S$, then storing $x$ in table $\C$ would imply that $\C(f)$ must be set to $1$, which is absurd. So, $x$ must be stored in table $\B$. \qed
	\end{proof}
	
	We next state the same observations in the context of table $\C$.

	\begin{observation}
		\label{obs:u2c1}
		Consider an element $e$ such that $\C(e) \cap \S = \{ e \}$, and suppose we want to store $e$ in table $\C$. If $f$ is a member of $\Uc(e)$ such that $\B(f) \cap \S = \{ f \}$, then all the other members of $\B(f)$ must be stored in table $\C$.
	\end{observation}

	\begin{observation}
		\label{obs:u2c2}
		Consider an element $e$ such that $\C(e) \cap \S = \{ e \}$, and suppose we want to store $e$ in table $\C$. If $f$ is a member $\Uc(e)$ such that $\B(f) \cap \S = \{ x \}$, where $x \neq f$, then $x$ must be stored in table $\C$.
	\end{observation}

\subsection{Bad Elements}

	We now define the notion of {\em good} and {\em bad} elements. These notions are motivated by the notions of {\em large} and {\em bounded sets} from Kesh~\cite{DBLP:conf/fsttcs/Kesh18}.

	\begin{definition}
		\label{def:bad}
		$e$ is a bad element w.r.t. table $\B$ if one of the following holds.
		\begin{enumerate}
			\item Two elements of $\Ub(e)$ share a set in table $\C$.
			\item The size of $\UB(e)$ is greater than $2s$, i.e. $|\UB(e)| > 2 \cdot s$.
		\end{enumerate}
		Otherwise, it is said to be good. Bad and good elements w.r.t. to table $\C$ are similarly defined. 
	\end{definition}
	
	The next claims state the consequences of an element being bad due to any of the above properties getting satisfied.

	\begin{claim}
		If two elements of $\Ub(e)$ share a set in table $\C$, then
		$\exists$ a subset $\S$ that contains $e$ and has size two such that to store $\S$, 
		$e$ cannot be stored in $\B$.
	\end{claim}

	\begin{proof}
		Suppose the elements $x, y \in \Ub(e)$ share the set $Y$ in
		table $\C$. We would prove that to store the subset $\S = \{ e,
		x \}$, $e$ cannot be stored in table $\B$.

		Let us say that $e$ indeed can be stored in table $\B$.
		According to Observation~\ref{obs:ub}, all elements of $\Ub(e)$,
		including $x$ and $y$, must be stored in table $\C$. As $x \in
		\S$, the bit corresponding to the set $Y$ must be set to $1$. As
		$y \notin \S$, the bit corresponding to $Y$ must be set to $0$.
		We thus arrive at a contradiction. So, $e$ cannot be stored in
		table $\B$. \qed
	\end{proof}

	\begin{claim}
	
		If the size of $\UB(e)$ is greater than $2s$, then $\exists$ a
		subset $\S$ that contains $e$ and has size at most three such that to
		store $\S$, $e$ cannot be stored in table $\B$.
		
	\end{claim}

	\begin{proof}
		Consider the set of those elements $f$ of $\UB(e)$ such that it is the only member of its block to belong to $\UB(e)$. As
		there are a total of $s$ blocks, there could be at most $s$ such
		elements. Removing those elements from $\UB(e)$ still leaves us
		with more than $s$ elements in $\UB(e)$. These remaining
		elements have the property that there is at least one other
		element from its block that is present in $\UB(e)$. Let this set
		be denoted by $Z$.

		As the size of $Z$ is larger than the size of table $\B$, there
		must exist at least two elements $x, y \in Z$ that share a set $X$
		in table $\B$. According to Definition~\ref{def:u2}, this
		implies that there exists elements $z, z' \in \Ub(e)$ such that
		$x \in \C(z) \setminus \{ z \}$ and $y \in \C(z') \setminus \{
		z' \}$. It might very well be that $z = z'$.
		
		If $x \in \A(e)$, as $e$ has been stored in table $\B$, so all the elements of $\A(e)$, including $x$, must have been stored in table $\B$. Consider the subset $\S = \{ e, x, z' \}$. As $x \in \S$, the bit corresponding to set $X$ must be set to $1$. As $e$ is stored in table $\B$, Observation~\ref{obs:ub} tells us that $z' \in \Ub(e)$ must be stored in table $\C$. As $\C(z') \cap \S = \{ z' \}$, Observation~\ref{obs:u2b1} tells us that $y$ must be stored in table $\B$. So, the bit corresponding to set $X$ must be set to $0$, which is absurd. So, to store $\S$, $e$ cannot be stored in table $\B$. This argument holds even if $x = e$. Similar is the case if $y \in \A(e)$. 
		
		If $x \in \B(e) \setminus \{ e \}$, and as $e$ has been stored in table $\B$, Observation~\ref{obs:ub} tells us that $x$ must be stored in table $\C$. Consider the subset $\S = \{ e, z \}$. Observation~\ref{obs:u2b1} tells us that as $z \in \Ub(e)$ is in $\S$, $x \in \C(z)$ cannot be stored in table $\C$, which is absurd. So, to store $\S$, $e$ cannot be stored in table $\B$. We can similarly argue the case $y \in \B(e)$.
		
		We now consider the case when $x, y \notin \A(e)$ and $\notin \B(e)$. If $\S$ contains $e$ and $x$, and we store $e$ in table $\B$, Observation~\ref{obs:ub} tells us that $z \in \Ub(e)$ must be stored in table $\C$, and as $x \in \C(z)$, Observation~\ref{obs:u2b2} tells us that $x$ must be stored in table $\B$. As $x \in \S$, hence the bit corresponding to set of $x$ in table $\B$, which is $X$, must be set to $1$.
		
		If $z \neq z'$, we include $z'$ in $\S$, and according to Observation~\ref{obs:u2b1}, $y \in \C(z')$ must be stored in table $\B$. As $y \in X$ is not in $\S$, $X$ must be set to $0$, and we arrive at a contradiction for the subset $\S = \{ e, x, z' \}$.
		
		It could also be the case that $z = z'$. As $y \in Z$, there exists an element $y' \in \A(y) \cap \UB(e)$. Let $y' \in C(z'')$, where $z'' \in \Ub(e)$. In this scenario, we consider storing the subset $\S = \{ e, x, z'' \}$. As, $z'' \in \S \cap \Ub(e)$, and $y' \notin \S$, Observation~\ref{obs:u2b1} implies that $y'$, and hence the whole of block $\A(y')$, including $y$, must be stored in table $\B$. As $y \notin \S$, the set of $y$ in table $\B$, which is $X$, must be set to $0$, and we again arrive at a contradiction.

		So, we conclude that $e$ in either of the cases cannot be stored
		in table $\B$. \qed
	\end{proof}

	The two claims above imply the following -- if an element $e$ is bad
	w.r.t. table $\B$ (Definition~\ref{def:bad}) due to Property 1, or if
	this property does not hold but Property 2 does, then there exists a
	subset, say $\S_1$, of size at most three containing $e$ such that to
	store $\S_1$, $e$ cannot be stored in table $\B$. The claims above also
	hold w.r.t. table $\C$. So, we can claim that if $e$ is bad w.r.t. to
	table $\C$, then there exists a subset $\S_2$ containing $e$ of size at
	most three such that to store $\S_2$, $e$ cannot be stored in table
	$\C$.

	Consider the set $\S = \S_1 \cup \S_2$. As $e$ is common in both the
	subsets, size of $\S$ is at most five. If $e$ is bad w.r.t. to table
	$\B$ and table $\C$, then to store subset $\S$, we cannot store $e$ in
	either of the tables, which is absurd. We summarise the discussion in the following
	lemma.

	\begin{lemma}
		\label{lem:bad}
		If an element $e$ is bad w.r.t. $\B$, then it must be good w.r.t
		$\C$.
	\end{lemma}

\subsection{Good Schemes}

	Based on the above lemma, we can partition our universe $\U$ into two
	sets $\U_1$ and $\U_2$ -- one that contains all the good elements w.r.t.
	to table $\B$, and one that contains the bad elements. We now partition
	each block and each set of the three tables of our datastructure into
	two parts, one containing elements from $\U_1$, and one containing the
	elements from $\U_2$. For elements of $\U_1$, only those blocks and sets
	that contain elements of $\U_1$ will be affected; similarly for the
	elements of $\U_2$.

	In effect, we have two independent schemes, one for $\U_1$ and one for
	$\U_2$. In the scheme for $\U_1$, all the elements in table $\B$ are
	good. In the scheme for $\U_2$, all the elements in table $\B$ are bad,
	and consequently, Lemma~\ref{lem:bad} tells us that all the elements of
	table $\C$ are good. In the scheme for $\U_2$, we now relabel the table
	$\B$ to $\C$ and relabel the table $\C$ to $\B$. To make the new scheme
	for $\U_2$ work, we now have to store $0$ in the blocks of table $\A$ for $\U_2$
	when earlier we were storing $1$, and have to store $1$ when earlier we
	were storing $0$.

	This change gives us a new scheme with two important properties -- the
	size of the datastructure has doubled from the earlier scheme, and all
	the elements in table $\B$ are now good.

	\begin{lemma}
		\label{lem:good}
		
		Given a $(5, m, s, 2)$-scheme, we can come up with a $(5, m, 2
		\times s, 2)$-scheme such that all the elements of $\U$ are good
		w.r.t. to table $\B$ in the new scheme.
	\end{lemma}

\subsection{Space Complexity}

	Consider a $(5, m, 3 \times s, 2)$-scheme all of whose elements are good w.r.t. table $\B$. The table sizes then are each equal to $s$. According to Lemma~\ref{lem:good}, the $2$-universe of each element w.r.t. to $\B$ will be at most $2s$. So, the sum total of all the $2$-universe sizes of all the elements is upper bounded by $m \times 2s$.
	
	We now consider how much each set of table $\C$ contribute to the total. From Definition~\ref{def:u2}, we have the following --
	\[
		\sum_{e \in \U} \mid \UB(e) \mid = \sum_{e \in \U} \mid \bigcup_{f \in \Ub(e)} \C(f) \setminus \{ f \} \mid = \sum_{e \in \U} \left( \sum_{f \in \Ub(e)} \mid \C(f) \setminus \{ f \} \mid \right).
	\]
	As all the elements are good, and hence for every element $e$, no two elements of $\Ub(e)$ share a set in table $\C$, we can thus convert the union in Definition~\ref{def:u2} to summation.

	Resolving the above equation, the details of which can be found in the~\ref{appn:lower}, we have
	\[
		\sum_{e \in \U} \mid \UB(e) \mid \ \geq \ c \cdot \frac{m^4}{s^3},
	\]
	for some constant $c$. The proof show that the minimum value is achieved when all the blocks and the sets in the three tables are of the same size, i.e. $m / s$. This combined with the upperbound for total sum of the sizes of all 2-universes gives us
	\begin{align*}
		c \cdot \frac{m^4}{s^3} \leq m \times 2s.
	\end{align*}
	Resolving the equation gives us
	\[
		s = \Omega(m^{3/4}).
	\]
	This bound applies to good schemes that respect the two assumptions declared at the beginning of this section.
	
	Suppose we have an arbitrary adaptive $(5, m, s, 2)$-scheme. If we want to make all the tables in this scheme of the same size, we can add extra bits which will make the size of the data structure at most $3 \cdot s$. So, we get a $(5, m, 3 \times s, 2)$-scheme that respect Assumption~1.
	
	In Kesh~\cite{DBLP:conf/fsttcs/Kesh18}, sets which have multiple elements from the same block were referred to as {\em dirty sets}. {\em Clean sets} were those which contain elements from distinct blocks. It was shown in Section~3 that any scheme with dirty sets can be converted into a scheme with only clean sets by using twice amount of space. Though the final claim was made in context of $n = 3$, but the proof applies to any $n$. So, we can now have a $(5, m, 6 \times s, 2)$-scheme that respects both of our assumptions.
	
	Such a scheme can be converted into a scheme with only good elements in table $\B$ by using twice the amount of space as before. We now have a $(5, m, 12 \times s, 2)$-scheme, where the table sizes are all $4s$, and we have shown that $4s = \Omega(m^{3/4})$.
	
	We summarise our discussion in the following theorem on the lower bound for two-adaptive bitprobes schemes storing five elements.
	\begin{theorem}
		\label{thm:lower}
		$s_A(5, m, 2) = \Omega(m^{3/4})$.
	\end{theorem}

\section{Our Scheme}

In this section, we will present a scheme which stores an arbitrary subset of size at most five from a universe of size $m$, and answers the membership queries in two adaptive bitprobes. This scheme improves upon the $\O(m^{10/11})$-scheme by the authors~\cite{DBLP:conf/walcom/BaigKS19}, and is fundamentally different from that scheme in the way that here block sizes are nonuniform, and any two blocks in table $\C$ share at most one bit. As per the convention of that scheme, we will use the label $T$ to refer to the table $\A$, $T_0$ to refer to the table $\B$, and $T_1$ to refer to the table $\C$.

\noindent {\bf Superblock : }
In this scheme, we use the idea in Kesh~\cite{DBLP:conf/cocoa/Kesh17} of mapping the elements of the universe on a square grid. Furthermore, we have used the idea of Radhakrishnan {\em et al.}~\cite{DBLP:conf/esa/RadhakrishnanRR01} to divide the universe into blocks and superblocks. Our scheme divides the universe of size $m$ into superblock of size $x^{2}zt$. Each superblock is made up of rectangular grids of size $t \times z$, and there are $x^{2}$ of them as shown in  Figure \ref{fig:structure}. Further, each integral point on a grid represents a unique element.

 \begin{figure}[t]
	\centering
	\captionsetup{justification=raggedright}  

	\begin{minipage}{.4\textwidth}
		\centering
		
		\includegraphics[width=\textwidth]{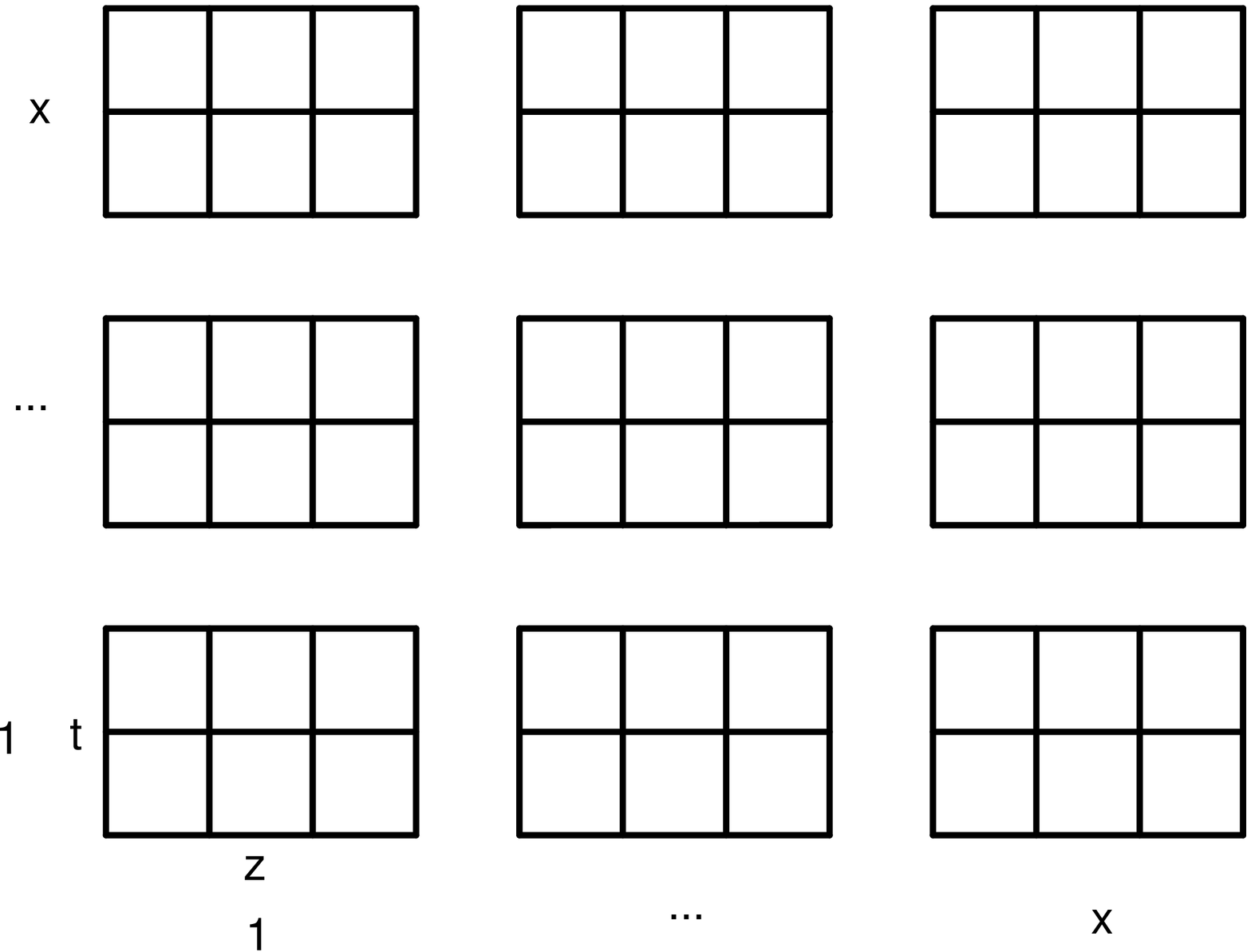}
		\caption{Figure showing structure of a superblock}
		\label{fig:structure}
	\end{minipage}%
\hspace{0.5cm}
	\begin{minipage}{.4\textwidth}
		\centering
		\vspace{-.3cm}
		\includegraphics[width=\textwidth]{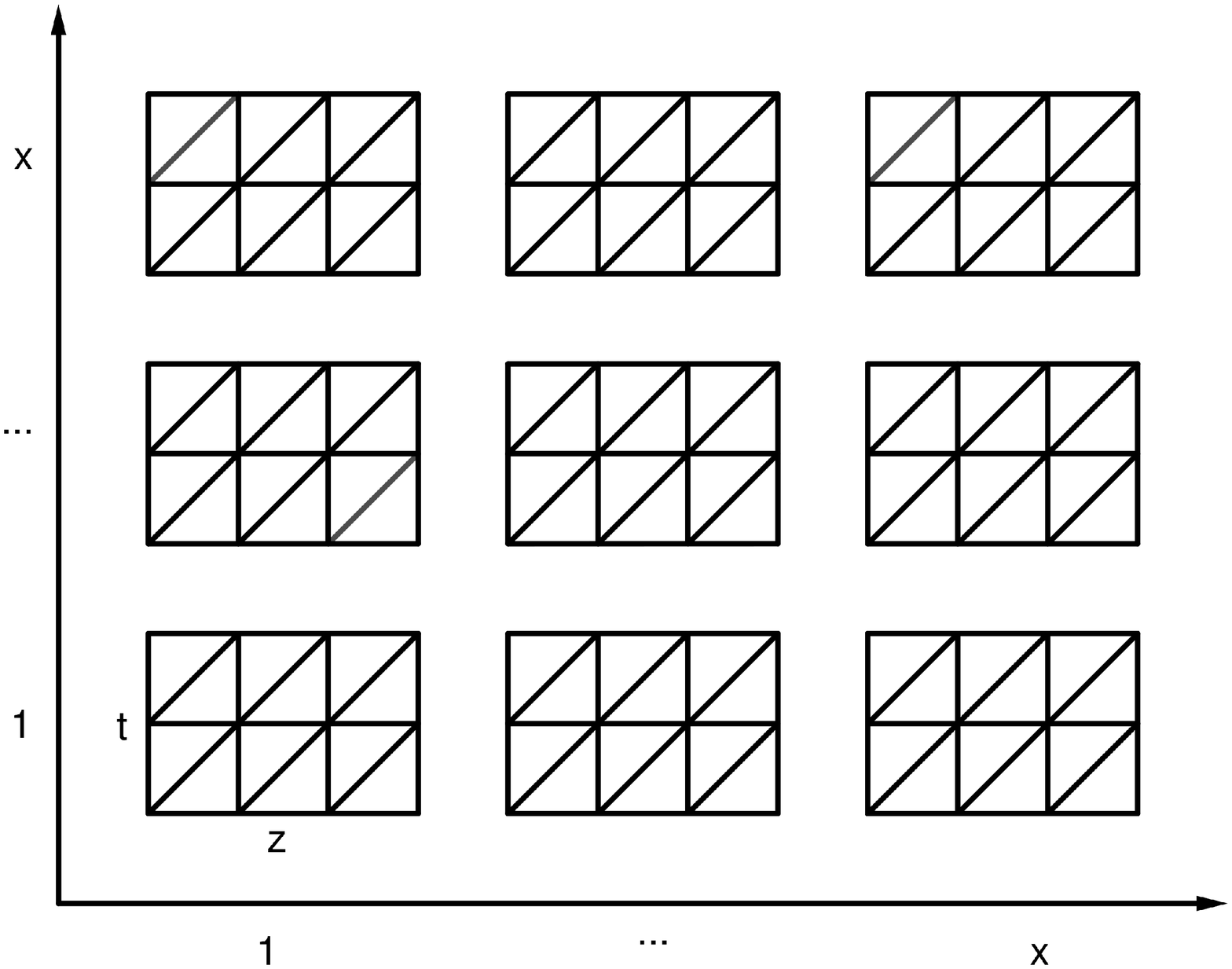}
		\caption{Lines drawn in the first superblock}
		\label{fig:line} 
	\end{minipage}	
\end{figure}

\noindent {\bf Block : }
For the $1$st superblock we draw lines with slope $1$ as shown in Figure \ref{fig:line}. Each line drawn represents a block. From Figure \ref{fig:line}, we can see that some blocks are of equal size and some are of different size. We do this for all the superblocks, and hence partitioning the universe into blocks. For the $i$th superblock we draw lines with slope $1/i$.

\noindent {\bf Table $T_1$ : } This table has space equal to that of a single superblock, i.e., $x^{2}zt$. All the superblocks can be thought of as superimposed over each other in this table. Structure of this table is shown in Figure \ref{fig:structure}.

\noindent {\bf Table $T$ : } In this table, we store a single bit for each block. Let there be $n$ superblocks in total. Now let us concentrate on a single grid of Figure \ref{fig:line}. The number of lines drawn for the $i$th superblock is equal to $z + c\cdot it$, where $c$ is a constant. If we sum this for all the superblocks total number of lines drawn for the single grid will be equal to $nz + c \cdot n^{2}t$. Now, since there are  $x \times x$ grids, the total number of lines drawn for all the superblocks will be $(nz + c \cdot n^{2}t)x^{2}$. As mentioned earlier, each line represents a block, and for each block, we have one bit of space in table $T$. So the size of this table $T$ is $(nz + c \cdot n^{2}t)x^{2}$ bits.

\noindent {\bf Table $T_0$ : }
In addition to lines drawn in superblocks to divide them into blocks, we also draw dotted lines in all the superblocks, as shown in Figure \ref{fig:dotted}. For the $i$th superblock we draw dotted lines with slope $1/i$. Further, we store a block of size $t$ in table $T_0$ for each dotted lines drawn. Now, we can see that for a specific superblock there could be many blocks belonging to that superblock which lies on the same dotted line. All the blocks which lie completely on the same dotted line query the same block in table $T_0$ kept for the dotted line.
\begin{wrapfigure}{r}{0pt}
	\centering
	\centering
	\vspace{-10pt}
	\includegraphics[width=0.5\textwidth]{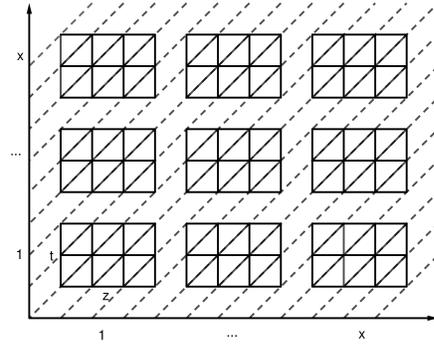}
	\caption{Dotted lines drawn for the first superblock}
	\label{fig:dotted}
	\vspace{-20pt}	

\end{wrapfigure}
Now let us talk about the space taken by table $T_0$. Using the idea shown in Figure \ref{fig:dotted} to draw the dotted lines, if we sum the total number of dotted lines drawn for all the superblocks which pass through x-axis, we will get $nzx$. Further, if we sum the total number of dotted lines drawn for all the superblocks from the y-axis, we get it to be less than or equal to $c_{1} \cdot n^{2}t \times x $, where $c_{1}$ is a positive constant. If we sum the total number of dotted lines drawn for all the superblocks from x and y-axis, we get $nzx + c_{1} \cdot n^{2}t \times x$. Since we store a block of size $t$ for each dotted line drawn, total space for table $T_0$ is $(nzx + c_{1} \cdot n^{2}t \times x) \times t$.

\noindent {\bf Size of data structure : }
Summing up the space taken by all the tables we get the following equation:
\begin{equation}
s(x,z,t)= x^{2}zt + (nz + c \cdot n^{2}t)x^{2} + (nzx + c_{1} \cdot n^{2}t \times x) \times t 
\end{equation} 
As mentioned earlier size of each superblock is $x^{2}zt$, so the total number of superblocks are $n= m/(x^{2}zt)$. Substituting this in the above equation, we get the following:
\begin{equation}
s(x,z,t)=  c_{1} \cdot \frac{m^{2}}{x^{3}z^{2}} + c \cdot \frac{m^{2}}{x^{2}z^{2}t} + \frac{m}{x} + \frac{m}{t} + x^{2}zt
\end{equation}
Choosing  $x = t = m^{1/6}$ and $z = m^{2/6}$, we get the space taken by our data structure to be $\O(m^{5/6})$.

\subsection{Query Scheme}
Our query scheme has three tables $T, T_0$ and $T_1$. Given a query element, we first find out the blocks to which it belongs. Further, we query the bit stored for this block in table $T$. If the bit returned is zero, we make the next query to table $T_0$ otherwise to table $T_1$. We say that query element is part of the set given to be stored if and only if last bit returned is one.

\subsection{Storage Scheme} \label{appn:upper} Our storage scheme sets the bits of tables $T, T_0$ and $T_1$ to store an arbitrary subset of size at most five in such a way that membership queries can be answered correctly. Storage scheme sets the bit of data structure depending upon the distribution of elements in various superblocks. Distribution of elements into various superblocks leads to various cases of the storage scheme.  While generating various cases we consider an arbitrary subset $S=\{n_1,n_2,n_3,n_4,n_5\}$ of size five given to be stored. Each block is either sent to  Table $T_0$ or $T_1$, and we store its bit-vector there. While sending blocks to either $T_0$ or $T_1$, we make sure that no two blocks sharing a bit have conflicting bit common in either of the tables, the correctness of the scheme relies on this fact. Keeping in mind the space constraints, we have discussed a few cases in this section, and for the sake of completeness the rest of the cases which can be handled in a similar fashion are mentioned in the~\ref{5:2upper}. Most of the cases generated and assignment made are similar to those generated in the previous paper on the problem by Baig {\em et al.}~\cite{DBLP:conf/walcom/BaigKS19} to store an arbitrary subset of size at most five.\\
\\
\textbf{Case 1 } All the elements belonging to $S$ belongs to the same superblock. In this case, we send all the blocks having elements from the set given to be stored to table $T_1$. All the empty blocks, i.e., blocks which do not have any elements from $S$ are sent to table $T_0$.\\
\\
\textbf{Case 2 } Four elements $S_1= \{n_1,n_2,n_3,n_4\}$ lies in one superblock and one $S_2=\{n_5\}$ in other. In this case, we send the block having element $n_5$ to table $T_1$ and rest all the blocks belonging to superblock which contains this element to table $T_0$. All the blocks which are having conflicting bit common with the block having element $n_5$ are sent to table $T_0$. Remaining all the blocks of superblocks which contains elements from $S_1$ are sent to table $T_1$. Furthermore, rest all the empty blocks of all the superblocks are sent to table $T_0$.\\
\\
 \textbf{Case 3 }
All the elements  $n_1,n_2,n_3,n_4$ and $n_5$ lies in the different  superblocks. In this case, we send all the blocks having elements to table $T_0$ and all the empty blocks to table $T_1$.

We conclude this section with the following theorem:
\begin{theorem}
\label{thm:upper}
There is a fully explicit two adaptive bitprobe scheme, which stores an arbitrary subset of size at most five, and uses $\mathcal{O}(m^{5/6})$ space. 	
\end{theorem}

\section{Conclusion}

In this paper, we have studied those schemes that store subsets of size at most five and answer membership queries using two adaptive bitprobes. Our first result improves upon the known lower bounds for the problem by generalising the notion of universe of sets in Kesh~\cite{DBLP:conf/fsttcs/Kesh18} to what may be referred to as {\em second order} universe. We hope that suitably defining still higher order universes will help address the lower bounds for subsets whose sizes are larger than five. Though the lower bound of $\Omega(m^{3/4})$ is an improvement, we believe that it is not tight.

We have also presented an improved scheme for the problem. It refines the approach taken by Baig {\em et al.}~\cite{DBLP:conf/walcom/BaigKS19} and alleviates the need for blocks that overlap completely to save space. This approach helps us achieve an upper bound of $\O(m^{5/6})$ which is a marked improvement over existing schemes.

\bibliographystyle{splncs04}
\bibliography{references}

\newpage
\appendix
\renewcommand\thesection{Appendix \Alph{section}}

\section{Lower Bound}
\label{appn:lower}

In this section, we prove our expression for the space lower bound. We start by proving a simple fact about sum of products.

\begin{claim}
	Given that $\sum_{i = 1}^n a_i \geq C_1$ and $\sum_{i = 1}^n b_i \geq C_2$, then
	\[
		\sum_{i = 1}^n a_i b_i \geq \frac{C_1 C_2}{n}.
	\]
\end{claim}

\begin{proof}
	Consider the following sum --
	\[
		\sum_{i = 1}^n (a_i + b_i)^2.
	\]
	This is minimised when the all of the summands are equal. Thus,
	\[
		\sum_{i = 1}^n (a_i + b_i)^2 \geq \sum_{i = 1}^n \left( \frac{C_1 + C_2}{n} \right)^2 = \frac{\left( C_1 + C_2 \right)^2}{n}.
	\]
	We can now expand the sum to prove the desired inequality.
	\begin{align*}
		\sum_{i = 1}^n (a_i + b_i)^2 &= \sum_{i = 1}^n a_i^2 + \sum_{i = 1}^n b_i^2 + \sum_{i = 1}^n 2 a_i b_i \\
		&\geq n \left( \frac{C_1}{n} \right)^2 + n \left( \frac{C_2}{n} \right)^2 + \sum_{i = 1}^n 2 a_i b_i \\
		&\geq \frac{\left( C_1 + C_2 \right)^2}{n} \\
		\implies \sum_{i = 1}^n 2 a_i b_i &\geq \frac{\left( C_1 + C_2 \right)^2}{n} - n \left( \frac{C_1}{n} \right)^2 - n \left( \frac{C_2}{n} \right)^2 \\
		&= 2 \frac{C_1 C_2}{n}.
	\end{align*}
	\qed
\end{proof}

We apply the claim above repeatedly to prove the following lemma. It is important to note that the sum is computed w.r.t. table $\B$, and in table $\B$ all the elements are good.

\begin{lemma}
	$\sum_{e \in \U} \mid \UB(e) \mid \ \geq \ c \cdot \frac{m^4}{s^3}$.
\end{lemma}

\begin{proof}
	We have the following expression for the sum of the sizes of all $2$-universes of all elements.
	\begin{align*}
		\sum_{e \in \U} \mid \UB(e) \mid & = \sum_{e \in \U} \mid \bigcup_{f \in \Ub(e)} \C(f) \setminus \{ f \} \mid \\
		&= \sum_{e \in \U} \left( \sum_{f \in \Ub(e)} \mid \C(f) \setminus \{ f \} \mid \right)
	\end{align*}
	We could convert the union in the expression above into the summation as no two elements of $\Ub(e)$ share a set. We can similarly expand $\Ub(e)$ from Definition~\ref{def:u}.
	\begin{align*}
		\sum_{e \in \U} \mid \UB(e) \mid &= \sum_{e \in \U} \left( \sum_{f \in \Ub(e)} \mid \C(f) \setminus \{ f \} \mid \right) \\
		&= \sum_{e \in \U} \left( \sum_{g \in \B(e) \setminus \{ e \}} \left( \sum_{f \in \A(g) \setminus \{ g \}} \left( \mid \C(f) \setminus \{ f \} \mid \right) \right) \right)
	\end{align*}
	
	We will first compute the value of the following expression.
	\[
	\begin{aligned}
		\sum_{e \in \U} |\B(e) \setminus \{ e \}| &= \sum_{e \in \U} \left( |\B(e)| - 1 \right) = \sum_{e \in \U} |\B(e)| - m \\
		&= \sum_{X \in \B} c_X |X| - m. &\text{(collecting over the sets of $\B$)}
	\end{aligned}
	\]
	Here, the sum of the coefficients $c_X$ is $m$, and the number of terms, which is same as the number of sets of $\B$, is $s$. Further, $\sum_{X \in \B} |X| = m$. So, applying the above claim,
	\begin{align*}
		\sum_{e \in \U} |\B(e) \setminus \{ e \}| &= \sum_{X \in \B} c_X |X| - m \\
		&\geq \frac{m \cdot m}{s} - m & \text{(sum of the sizes of the sets of $\B$ is $m$)} \\
		&\geq c \frac{m^2}{s} & \text{(for some suitable coefficient c)}
	\end{align*}
	Next, we compute an expression the sum of whose coefficients is the above sum.
	\begin{align*}
		\sum_{e \in \U} \left( \sum_{g \in \B(e) \setminus \{ e \}} |\A(g) \setminus \{ g \}| \right) &\geq c \sum_{Z \in \A} c_Z |Z|, &\text{(collecting over the blocks of $\A$)} \\
		&\geq c' \cdot \frac{\frac{m^2}{s} m}{s} = c' \frac{m^3}{s^2}. & \text{(for some suitable coefficient c')}
	\end{align*}
	We finally compute the desired expression of which the sum of coefficients is the above expression.
	\begin{align*}
		\sum_{e \in \U} \mid \UB(e) \mid &= \sum_{e \in \U} \left( \sum_{g \in \B(e) \setminus \{ e \}} \left( \sum_{f \in \A(g) \setminus \{ g \}} \left( \mid \C(f) \setminus \{ f \} \mid \right) \right) \right) \\
		&\geq c \cdot \sum_{Y \in \C} c_Y |Y|, &\text{(collecting over the sets of $\C$)} \\
		&= c' \cdot \frac{\frac{m^3}{s^2} m}{s} = c' \frac{m^4}{s^3}. & \text{(for some suitable coefficient c')}
	\end{align*}
	\qed

\end{proof}

\section{Storage Scheme}
\label{5:2upper}
Rest of the case of Section~\ref{appn:upper} is discussed here.\\
\\
\textbf{Case 4 } Three elements $S_1=\{n_1,n_2,n_3\}$ belong to one superblock and two elements $S_2=\{n_4,n_5\}$ to the other.\\
\\
\textbf{Case 4.1 }
All the blocks to which elements from $S_1$ belong lies on the same dotted line of their superblock.\\
\\		
\textbf{Case 4.1.1 }		
Two blocks to which elements from $S_2$ belong have a conflicting bit common with the blocks corresponding to the elements from $S_1$ in table $T_1$. In this case, we send the blocks having elements from $S_2$ to table $T_0$. Further, we send empty blocks lying on the dotted lines to which blocks having elements from $S_2$ belongs to table $T_1$. We send all the blocks which contain elements from $S_1$ in table $T_1$. We send the rest of the empty blocks to table $T_0$.\\  
\\		 		
\textbf{Case 4.1.2 }	
Only one block which contains an element from $S_2$ has a conflicting bit common with the block corresponding to the elements from $S_1$ in table $T_1$. In this case, we send all the blocks which contain elements from $S_1$ to table $T_1$. We send the block having an element from $S_2$, and having conflicting bit common with block having an element from $S_1$, to table $T_0$, and the rest of the blocks which lies on the dotted line containing this block to table $T_1$. If after this other nonempty block having an element from $S_2$ is still unassigned then we send it to table $T_1$, and all the empty blocks lying on the dotted line containing this block to table $T_0$ . Rest all the empty blocks are sent to table $T_0$.\\   
\\		
\textbf{Case 4.1.3 }
None of the blocks which contain an element from $S_2$ have a conflicting bit common with the block which includes an element from $S_1$ in table $T_1$ . In this case, we send all the nonempty blocks to table $T_1$ and all the empty blocks to table $T_0$.\\  
\\			
\textbf{Case 4.2 }
Two blocks which contain elements say $n_1$ and $n_2$ from $S_1$ lies on the same dotted line and other say $n_3$ lies on a different dotted line.\\
\\			 
\textbf{Case 4.2.1 }
All the blocks which contain elements from  $S_2$ have a conflicting bit common with blocks which include elements from $S_1$ in table $T_1$.\\
\\
\textbf{4.2.1.1 }
Let us first consider the case where blocks having elements from $S_2$ have a conflicting bit common with the blocks having elements $n_1$ and $n_2$. In this case, we send the blocks having element $n_4$ and $n_5$ to  table $T_0$, and all the blocks lying on the dotted lines containing these block to table $T_1$. Further, we send the blocks having elements $n_1$ and $n_2$ to table $T_1$. Block having element $n_3$ is sent to table $T_0$, and all the empty blocks lying on the dotted line containing this block is sent to table $T_1$. Rest all the empty blocks are sent to table $T_0$.\\
\\     
\textbf{4.2.1.2 }
Without loss of generality let us now consider the case where blocks having an element from $S_2$ have a conflicting bit common with blocks having element $n_1$ and $n_3$. In this case, we send the blocks having elements $n_4$ and $n_5$ to table $T_1$, rest all the blocks lying on the dotted line(lines) containing these blocks to table $T_0$. Further, we send the blocks having element $n_1$ and $n_3$ to table $T_0$, and all the blocks lying on the dotted lines containing these blocks to table $T_1$. Rest all the blocks are sent to table $T_0$.\\
\\
\textbf{4.2.1.3}
Now we are left with a case where block having an element from $S_2$ have a conflicting bit common with blocks having elements $n_1, n_2$ and $n_3$. In this case, we send the blocks having element $n_4$ and $n_5$ to table $T_0$, and all the empty blocks lying on the dotted lines containing this block to table $T_1$. Further, we send all the blocks having elements from $S_1$	to table $T_1$, and rest all the empty blocks to table $T_0$.\\
\\		 
\textbf{Case 4.2.2 }
Only one block having an element from $S_2$ have a conflicting bit common with the block(blocks) having an element(elements) from $S_1$.\\
\\
\textbf{Case 4.2.2.1 } All the blocks having elements from $S_2$ lies on the same dotted line. Without loss of generality, let us say block having element $n_4$ from $S_2$ have a conflicting bit common with a block having an element from $S_1$. In this case, we send the blocks having element $n_4$ and $n_5$ to table $T_1$, and all the empty blocks lying on the dotted line containing these blocks to table $T_0$. Further, we send the block(blocks) having an element(elements) from $S_1$, and having a conflicting bit common with a block having element $n_4$  to table $T_0$, and all the blocks lying on the dotted line containing this block to table $T_1$. We now send the rest of the block(blocks) having an element from $S_1$ to table $T_1$. Rest all the empty blocks are sent to table $T_0$.\\
\\
\textbf{Case 4.2.2.2 } Now let us consider a case where blocks having an element from $S_2$ lies on the different dotted line. Without loss of generality lets say block having element $n_4$ have a conflicting bit common with a block(blocks) having an element(elements) from $S_1$. 

Let us first consider the case where block having element $n_4$ have a conflicting bit common with either block having element $n_1$ or $n_2$. Without loss of generality, let us say block having element $n_4$ have a conflicting bit common with a block having element $n_1$. In this case, we send the block having element $n_3,n_4$ and $n_5$ to table $T_0$, and all the empty blocks lying on the dotted line containing these blocks to table $T_1$. Now, we see the positions of the blocks having elements $n_1$ and $n_2$. Let us first consider the case where blocks having element $n_1$ or $n_2$ have conflicting bit common with one of the empty blocks lying on the dotted line which contains block having element $n_5$. Without loss of generality let us say that block having element $n_1$ have conflicting bit common with one of the empty blocks lying on the dotted line which contains block having element $5$. In this case, we send the block having element $n_1$ to table $T_0$, and rest all the blocks lying on the dotted line containing this block to table $T_1$. Rest all the empty blocks are sent to table $T_0$. On the other hand, if the block having element $n_1$ or $n_2$ do not have conflicting bit common with empty blocks lying on the dotted line which contains block having element $n_5$, then we send the blocks having elements $n_1$ and $n_2$ to table $T_1$, and rest all the empty blocks lying on the dotted line containing these blocks to the table $T_0$. Rest all the empty blocks are sent to the table $T_0$.

Now let us consider the case where block having element $n_4$ have a conflicting bit common with a block having element $n_3$. Now, we see if the block having element $n_4$ have conflicting bit common with block having element $n_1$ or $n_2$. Without loss of generality, let us say that block having element $n_4$ have conflicting bit common with block having element $n_1$. In this case, we can use the assignment made in previous paragraph. On the other hand, if the block having element $n_4$ do not have conflicting bit common with block having element $n_1$ or $n_2$, then we send the blocks having elements $n_3$ and $n_5$ to table $T_0$, and all the empty blocks lying on the dotted lines containing these blocks to table $T_1$. Further, we send the block having element $n_4$ to table $T_1$, and all the empty blocks lying on the dotted line containing this block to table $T_0$. Now we see the position of the blocks having elements $n_1$ and $n_2$. Let us first consider the case where on of the blocks having elements $n_1$ or $n_2$ has conflicting bit common with one of the empty blocks lying on the dotted line which contains block having element $n_5$. Without loss of generality, let us say that block having element $n_1$ has conflicting bit common with one of the empty blocks lying on the dotted line which contains block having element $n_5$. In this case, we send the block having element $n_1$ to table $T_0$, and all the blocks lying on the dotted line containing this block to table $T_1$. Rest all the empty blocks are sent to table $T_0$. On the other hand, if none of the block having element $n_1$ or $n_2$ has conflicting bit common with empty blocks lying on the dotted line which contains block having element $n_5$, then we send the block having element $n_1$ and $n_2$ to table $T_1$, and rest all the empty blocks to table $T_0$.\\
\\
\textbf{Case 4.2.3 }
None of the blocks which contain an element from $S_2$ have a conflicting bit common with the block which includes an element from $S_1$ in table $T_1$ . In this case, we send all the nonempty blocks to table $T_1$ and all the empty blocks to table $T_0$.\\
\\
\textbf{Case 4.3 }All the blocks having an element from $S_1$ lies on the different dotted lines.\\ 
\\
\textbf{Case 4.3.1 }
Both the blocks having elements $n_4$ and $n_5$ have a conflicting bit common with the blocks having elements from $S_1$ in table $T_1$. Now we see the positions of the blocks having elements $n_4$ and $n_5$. Blocks having elements $n_4$ and $n_5$ can either lie on the same dotted line or on the different dotted lines. If the blocks having elements $n_4$ and $n_5$ lie on the different dotted lines, then we send all the blocks having elements to table $T_0$, and all the empty blocks to table $T_1$. Now we consider the case where blocks having elements $n_4$ and $n_5$ lies on the same dotted line. Furthermore, without loss of generality let us consider that blocks having elements $n_4$ and $n_5$ conflicts with blocks having elements $n_1$ and $n_2$ respectively. In this case, we send the blocks having elements $n_1$ and $n_2$ to table $T_0$, and all the empty blocks lying on the dotted lines containing these blocks to the table $T_1$. Further, we send the blocks having elements $n_4$ and $n_5$ to table $T_1$, and all the empty blocks lying on the dotted line containing these blocks to table $T_1$. Now, we see the position of the block having element $n_3$. If the block having element $n_3$ has conflicting bit common with blocks having elements $n_4$ or $n_5$, then we send the block having element $n_3$ to table $T_0$, and all the empty blocks lying on the dotted line which contains this block to table $T_1$. Rest all the empty blocks to table $T_0$. On the other hand, if the block having element $n_3$ do not have conflicting bit common with block having element $n_4$ or $n_5$, then we send the  block having element $n_3$ to table $T_1$, and rest all the empty blocks to table $T_0$.   \\
\\	
\textbf{Case 4.3.2 }
Only one of the block having element say $n_4$ from $S_2$ have a conflicting bit common with blocks having an element from $S_1$ in table $T_1$. Similar to the last case, in this case also we see whether blocks having elements
$n_4$ and $n_5$ lie on the same dotted line or on the different dotted lines. If the blocks having elements $n_4$ and $n_5$ lies on the different dotted lines, then we send the blocks having elements to table $T_0$, and all the empty blocks to table $T_1$. Now, we consider the case where blocks having elements $n_4$ and $n_5$ lies on the same dotted lines.  Furthermore, without loss of generality let us say that block having the element $n_1$ have a conflicting bit common with the block having the element $n_4$. In this case, we send the blocks having elements $n_1$ and $n_4$ to table $T_0$, and all the empty blocks lying on the dotted lines containing these blocks to table $T_1$. Further, we send the blocks having elements $n_2,n_3$ and $n_5$ to table $T_1$. Rest all the empty blocks are sent to table $T_0$.\\   	
\\		
\textbf{Case 4.3.3 }
None of the blocks which contain an element from $S_2$ have a conflicting bit common with the block which includes an element from $S_1$. This case is the same as Case 4.1.3.\\
\\
\textbf{Case 5}
The elements in $S_1 =\{n_1,n_2,n_3\}$  lies in a superblock and the elements $n_4$ and $n_5$ in the different superblocks.\\ 
\\	   		
\textbf{Case 5.1}
Blocks having element $n_1,n_2$ and $n_3$ lies on a same dotted line.\\
\\	   			
\textbf{Case 5.1.1}
Blocks having element $n_4$ and $n_5$ have a conflicting bit common with the blocks having elements from $S_1$ in table $T_1$. In this case, we send the blocks having elements $n_4$ and $n_5$ to table $T_0$ and the rest of the empty blocks which lie on the dotted lines containing these blocks to table $T_1$. We send the blocks having elements from $S_1$ to table $T_1$. Rest all the empty blocks are sent to table $T_0$.\\   
\\	   			
\textbf{Case 5.1.2}
Only one of the block having element say $n_4$ have a conflicting bit common with the block having an element from $S_1$ in table $T_1$. Without loss of generality let us say block having the element $n_1$ have a conflicting bit common with the block having the element $n_4$. In this case, we send the blocks having elements $n_4$ and $n_5$ to table $T_0$ and all the empty blocks lying on the dotted line containing these blocks to table $T_1$. Now we see whether the dotted line which contains block having the element $n_5$ passes through the block having element elements from $S_1$ or not. Let us first consider a case where the dotted line which contains block having the element $n_5$ passes through one of the blocks having elements from $S_1$, without loss of generality let us say it passes through block having the element $n_2$. In this case, we send the block having the element $n_2$ to table $T_0$ and rest all the blocks lying on the dotted line containing this block to table $T_1$. Rest all the empty blocks are sent to table $T_0$. On another hand, if the dotted line which contains block having the element $n_5$ does not pass through any of the block having element from $S_1$ then we send the blocks having elements from $S_1$ to table $T_1$. Rest all the empty blocks are sent to table $T_0$.  \\ 
\\
\textbf{Case 5.1.3}
None of the blocks which contains element  $n_4$ or $n_5$ have a conflicting bit common with the blocks which include an element from $S_1$ in table $T_1$. Now, the blocks having elements $n_4$ and $n_5$ can conflict among themselves or it does not. If the block having elements $n_4$ and $n_5$ do not conflict among themselves, then we send all the blocks having elements to table $T_1$, and rest all the empty blocks to table $T_0$. On the other hand, if the blocks having elements $n_4$ and $n_5$ conflict among themselves, then we see whether they conflict on the dotted line which contains block having elements from $S_1$. If they conflict on the dotted line which contains block having elements from $S_1$, then we send the blocks having elements $n_4$ and $n_5$ to table $T_0$, and all the empty blocks lying on the dotted line containing this block to table $T_1$. Further, we send the blocks having elements from $S_1$ to table $T_1$, and rest all the empty blocks to table $T_0$. If the blocks having elements $n_4$ and $n_5$ do not conflict on the dotted line which contains blocks having elements from $S_1$,   then we send the block having element $n_4$ to table $T_0$, and all the empty blocks lying on the dotted line containing this block to table $T_1$. Further, we send the block having element $n_5$ to table $T_1$, and all the empty blocks lying on the dotted line containing this block to table $T_0$. Now, we see whether the dotted line which contains block having element $n_4$ passes through block having element $n_1$ or $n_2$. Without loss of generality, let us say that the dotted line which contains block having element $n_4$ passes through block having element $n_1$. In this case, we send the block having element $n_1$ to table $T_0$, and rest all the blocks lying on the dotted line containing this block to table $T_1$. Rest all the empty blocks are sent to table $T_0$. If the dotted lines which contains block having elements $n_4$ or $n_5$ do not pass through block having elements from $S_1$, then we send the blocks having elements from $S_1$ to table $T_1$, and rest all the empty blocks to table $T_0$. \\
\\	   			
\textbf{Case 5.2}
Two elements say $n_1$ and $n_2$ lies on the same dotted line and the element $n_3$ lies on a different dotted line.\\
\\	   				
\textbf{Case 5.2.1}
Blocks having the elements $n_4$ and $n_5$ have a conflicting bit common with the blocks having elements from $S_1$ in table $T_1$.\\
\\
\textbf{Case 5.2.1.1}
Blocks having the elements $n_4$ and $n_5$ have a conflicting bit common with the block having elements $n_1$ and $n_2$ in table $T_1$. In this case, we send the blocks having elements $n_4$ and $n_5$ to table $T_0$ and all the empty blocks lying on the dotted lines containing these blocks to table $T_1$. Also, we send the blocks having elements $n_1$ and $n_2$ to table $T_1$. We send the block which contains the element $n_3$ to table $T_0$ and the rest of the empty block lying on the dotted line containing this block to table $T_1$. Rest all the empty blocks are sent to table $T_0$.\\  
\\	   					
\textbf{Case 5.2.1.2}
Blocks having elements $n_4$ and $n_5$ have a conflicting bit common with the blocks lying on the different dotted lines, say block having elements $n_1$ and $n_3$ in table $T_1$. In this case, we send the blocks having elements $n_4$ and $n_5$ to table $T_1$ and the rest of the empty blocks lying on the dotted line containing these blocks to table $T_0$ . We send the blocks having elements $n_1$ and $n_3$ to table $T_0$ and the rest of the blocks lying on these dotted lines to table  $T_1$. Rest all the empty blocks are sent to table $T_0$.\\  
\\	   					
\textbf{Case 5.2.1.3}
Blocks having element $n_4$ and $n_5$ have a conflicting bit common in table $T_1$. If these blocks have a conflicting bit common with the block having the element $n_1$ or $n_2$ then we send both the blocks having the element $n_4$ and $n_5$ to table $T_0$ and the rest of the empty block lying on the dotted line containing these blocks to table $T_1$. Also, we send the blocks having the element $n_1, n_2$  and the blocks lying on the dotted line containing these blocks to table $T_1$. We send the block which contains element $n_3$ to table $T_0$ and the rest of the empty block which lies on the dotted line containing this block to table $T_1$. Rest all the empty blocks are sent to table $T_0$.

If the blocks containing elements $n_4$ and $n_5$ have a conflicting bit common with the block which contains element $n_3$ in table $T_1$ then we send the block containing element $n_3$ to table $T_0$ and the rest of the empty block which lie on the dotted line containing this block to table $T_1$. Also, we send the block containing $n_4$ to table $T_1$ and the rest of the empty block lying on the dotted line containing this block to table $T_0$. We send the block having the element $n_5$ to table $T_0$ and the rest of the empty block lying on the dotted line containing this block to table $T_1$. Now we see whether the dotted line which contains block having the element $n_5$ passes through the block having the element  $n_1$ or $n_2$. Without loss of generality let us first consider the case where the dotted line which contains block having the element $n_5$ passes through the block having the element $n_1$. In this case, we send the block having the element $n_1$ to table $T_0$ and all the blocks lying on the dotted line containing this block to table $T_1$. Rest all the empty blocks are sent to table $T_0$. If the dotted line which contains block having element $n_5$ does not pass through the blocks having elements $n_1$ or $n_2$ then we send the blocks having elements $n_1$ and $n_2$ to table $T_1$ and rest all the empty blocks to table $T_0$.\\ 
\\	   					
\textbf{Case 5.2.2}
Only one block having an element $n_4$ or $n_5$, have a conflicting bit common with a block having the element from $S_1$ in table $T_1$. Let us first consider the case where block having element $n_4$ have a conflicting bit common with the block having element $n_1$ or $n_2$. Without loss of generality let us say block having the element $n_4$ have a conflicting bit common with the block having the element $n_1$. Now we can have two cases, either the block having element $n_4$ conflicts with block having the element $n_3$ or it does not. Let us first consider the case where block having the element $n_4$ do not conflict with block having the element $n_3$. In this case, we send the blocks having the elements $n_3, n_4$ and $n_5$ to table $T_0$, and all the empty blocks lying on the dotted line containing these blocks to table $T_1$. Now we see whether the empty blocks lying on the dotted line which contains block having element $n_5$ passes through blocks having elements $n_1$ or $n_2$. Without loss of generality let us say that empty block lying on the dotted line which contains block having element $n_5$ passes through block having element $n_1$. In this case, we send the block having the element $n_1$ to table $T_0$ and rest all the blocks lying on the dotted line containing this block to table $T_1$. Rest all the empty blocks are sent to table $T_0$. On the other hand, if the dotted line which contains block having element $n_5$ do not pass through blocks having element $n_1$ or $n_2$, then we send the block having element $n_1$ and $n_2$ to $T_1$, and rest all the empty blocks to table $T_0$. Now we consider the case where block having element $n_4$ have conflicting bit common with block having element $n_3$. In this case, we send the block having element $n_4$ to table $T_0$, and all the empty blocks lying on the dotted line containing this block to table $T_1$. Now we see the position of the block having element $n_5$. If the block having element $n_5$ have conflicting bit common with empty block lying on the dotted line containing blocks having elements $n_1$ and $n_2$, then we send the block having element $n_1, n_2, n_3$ and $n_5$ to table $T_1$, and rest all the empty blocks to table $T_0$. If the block having element $n_5$ do not have conflicting bit common with empty block lying on the dotted line containing blocks $n_1$ and $n_2$, then we send the block having element $n_5$ to table $T_0$, and all the empty blocks lying on the dotted line containing this block to table $T_1$. Now we see whether the blocks having element $n_1$ or $n_2$ have conflicting bit common with empty block lying on the dotted line containing block having element $n_5$. Without loss of generality let us say that block having element $n_1$ have conflicting bit common with empty block lying on the dotted line containing block having element $n_5$. In this case, we send the block having element $n_1$ to table $T_0$, and rest all the blocks lying on the dotted line containing this block to table $T_1$. Further, we send the block having $n_3$ to table $T_0$, and all the empty blocks lying on the dotted line containing this block to table $T_1$. Rest all the empty blocks are sent to table $T_0$. If none of the blocks having element $n_1$ or $n_2$ have conflicting bit common with empty block lying on the dotted line containing element $n_5$, then we send the blocks having elements $n_3$ and $n_5$ to table $T_0$, and all the empty blocks lying on the dotted line containing this block to table $T_1$. Furthermore, we send the blocks having element $n_1$ and $n_2$ to table $T_1$, and rest all the empty blocks to table $T_0$. 

Now let us consider the case where block having the element $n_4$ have a conflicting bit common only with the block having the element $n_3$. Now we see whether the blocks having elements $n_4$ and $n_5$ conflicts or not. Let us first consider the case where blocks having elements $n_4$ and $n_5$ do not conflicts. Now we see the position of the block having element $n_5$. Let us first consider the case where block having element $n_5$ conflicts with empty block lying on the dotted line which contains block having element $n_1$. In this case, we send the block having the element $n_1, n_2, n_4$ and $n_5$ to table $T_1$, and rest all the empty blocks lying on the dotted lines containing these blocks to table $T_0$. We send the block having the element $n_3$ to table $T_0$, and all the empty blocks lying on the dotted line containing this block to table $T_1$. Rest all the empty blocks are sent to table $T_0$. If the block having the element $n_5$ have a conflicting bit common with a empty block lying on the dotted line which contains block having the element $n_3$, then we send the blocks having the elements $n_3$ and $n_5$ to table $T_1$, and all the empty blocks lying on the dotted lines containing these blocks to table $T_0$. We send the block having element $n_4$ to table $T_0$, and rest all the empty blocks lying on the dotted line  containing this block to table $T_1$. If the block having the element $n_1$  have a conflicting bit common with block lying on the dotted line which contains block having the element $n_4$, then we send the block having the element $n_1$ to table $T_0$, and rest all the blocks lying on the dotted line containing this block to table $T_1$. Rest all the empty blocks are sent to table $T_0$. Similar is the case if the block having the element $n_2$ have a conflicting bit common with block lying on the dotted line which contains block having the element $n_4$. On the other hand, if the dotted line which contains block having element $n_4$ do not pass through blocks having elements $n_1$ or $n_2$, then we send the blocks having elements $n_1$ and $n_2$ to table $T_1$, and rest all the empty blocks to table $T_0$. If none of the above case occurs, and the block having the element $n_5$ does not have a conflicting bit common with a block lying on the dotted line which contains block having the element $n_3$, then we send the block having the element $n_1, n_2, n_4$ and $n_5$ to table $T_1$, and rest all the empty blocks lying on the dotted lines containing these blocks to table $T_0$. Further, we send the block having element $n_3$ to table $T_0$, and all the empty blocks lying on the dotted line containing this block to table $T_1$. Rest all the empty blocks are sent to table $T_0$.  Now, we see the case where block having element $n_4$ and $n_5$ conflicts. Now we can have several cases depending upon whether blocks having element $n_4$ and $n_5$ conflicts with empty block lying on the dotted line which contains block having element $n_1$. Let us first consider the case where block having element $n_5$ conflicts with empty block lying on the dotted line which contains block having element $n_1$. In this case, we send the blocks having elements $n_3, n_4$ and $n_5$ to table $T_0$, and rest all the empty blocks lying on the dotted lines containing these blocks to table $T_1$. Now we see whether the dotted line which contains block having element $n_4$ passes through block having element $n_1$ or $n_2$. Without loss generality let us say that dotted line which contains block having element $n_4$ passes through block having element $n_1$. In this case, we send the block having the element $n_1$ to table $T_0$, and rest all the blocks lying on the dotted line containing this block to table $T_1$. Rest all the empty blocks are sent to table $T_0$. On the other hand, if the dotted line which contains block having element $n_4$ do not pass through blocks having elements $n_1$ or $n_2$, then we send the blocks having elements $n_1$ and $n_2$ to table $T_1$, and rest all the empty blocks to table $T_0$. Similar is the case when block having element $n_4$ conflicts with empty block lying on the dotted line which contains block having element $n_1$. If none of the above occurs, then we send the blocks having elements $n_3$ and $n_5$ to table $T_0$, and rest all the empty blocks lying on the dotted lines containing these blocks to table $T_1$. Further, we send the block having element $n_4$ to table $T_1$, rest all the empty blocks lying on the dotted line containing this block to table $T_0$. Now we see whether the dotted line which contains block having element $n_5$ passes through block having element $n_1$ or $n_2$. Without loss generality let us say that dotted line which contains block having element $n_5$ passes through block having element $n_1$. In this case, we send the block having the element $n_1$ to table $T_0$, and rest all the blocks lying on the dotted line containing this block to table $T_1$. Rest all the empty blocks are sent to table $T_0$. On the other hand, if the dotted line which contains block having element $n_4$ do not pass through blocks having elements $n_1$ or $n_2$, then we send the blocks having elements $n_1$ and $n_2$ to table $T_1$, and rest all the empty blocks to table $T_0$.    \\  
\\	   					
\textbf{Case 5.2.3}	
None of the blocks having the elements $n_4$ or $n_5$ have a conflicting bit common with blocks having elements from $S_1$ in table $T_1$. We can have several cases depending upon whether the blocks having element $n_4$ and $n_5$ conflicts or not. Let us first consider the case where blocks having elements $n_4$ and $n_5$ do not conflict. In this case we send all the blocks having elements to table $T_1$, and all the empty blocks to table $T_0$. Now let us consider the case where blocks having elements $n_4$ and $n_5$ conflicts. Now, we see the position of the block having elements $n_4$ and $n_5$. Let us first consider the case where blocks having elements $n_4$ and $n_5$ conflicts on the dotted line which contains block having element $n_1$. In this case, we send the blocks having elements $n_1$ and $n_2$ to table $T_1$, and rest all the empty blocks lying on the dotted line containing these blocks to table $T_0$. Further, we send the blocks having elements $n_3, n_4$ and $n_5$ to table $T_0$, rest all the empty blocks lying on the dotted lines containing these blocks to the table $T_1$. Rest all the empty blocks are sent to table $T_0$. Now let us consider the case where block having elements $n_4$ and $n_5$ conflicts outside the dotted line which contains block having element $n_1$. Now we can have a case where either block having elements $n_4$ or $n_5$ conflicts with empty block lying on the dotted line which contains block having element $n_1$. Without loss of generality let us say that block having element $n_4$ conflicts with empty block lying on the dotted line which contains block having element $n_1$. In this case, we send the blocks having elements $n_3, n_4$ and $n_5$ to table $T_0$, and rest all the empty blocks lying on the dotted line containing these blocks to table $T_1$. Now we see whether the dotted line which contains block having element $n_5$ passes through block having element $n_1$ or $n_2$. Without loss of generality, let us say that the dotted line which contains block having element $n_5$ passes through block having element $n_1$. In this case, we send the block having element $n_1$ to table $T_0$, and rest all the blocks lying on the dotted line containing this block to table $T_1$. Rest all the empty blocks are sent to table $T_0$. On the other hand, if the dotted line which contains block having element $n_5$ do not pass through block having element $n_1$ or $n_2$, then we send the block having element $n_1$ and $n2$ to table $T_1$, and rest all the empty blocks to table $T_0$. Now we consider the case where block having element $n_4$ or $n_5$ do not conflict with empty block lying on the dotted line which contains block having element $n_1$. In this case, we send the blocks having elements $n_3$ and $n_5$ to table $T_1$, and rest all the empty blocks lying on the dotted lines containing these blocks to table $T_0$. Further, we send the block having element $n_4$ to table $T_0$, and rest all the empty blocks lying on the dotted line containing this block to table $T_1$. Now we see if the dotted line which contains block having element $n_4$ passes through block having element $n_1$ or $n_2$. Without loss of generality let us say that dotted line which contains block having element $n_4$ passes through block having element $n_1$. In this case, we send the block having element $n_1$ to $T_0$, and rest all the blocks lying on the dotted line containing this block to table $T_1$. Rest all the empty blocks are sent to table $T_0$. If the dotted line which contains block having element $n_4$ do not pass through block having element $n_1$ or $n_2$, then we send the blocks having elements $n_1$ and $n_2$ to table $T_1$, and rest all the empty blocks to table $T_0$.   \\	
\\	   			
\textbf{Case 5.3}	
All the elements belonging to $S_1$ lies on the different dotted lines.  In this case, we send all the blocks having elements to table $T_0$ and all the empty blocks to table $T_1$.\\
\\	   				
\textbf{Case 6}
Two elements $S_1=\{n_1,n_2\}$ lies in a superblock other two elements $S_2=\{n_3,n_4\}$ lies in other superblock and an element $S_3=\{n_5\}$ in a different superblock.\\
\\	
\textbf{Case 6.1}
Blocks having elements belonging to $S_1$ lies on the same dotted line, blocks having elements belonging to $S_2$ lies on the same dotted line.\\
\\		
\textbf{Case 6.1.1}
Two blocks having elements from $S_2$ and $S_3$ have a conflicting bit common with the blocks having elements from $S_1$ in table $T_1$. Without loss of generality, let us say that block having element $n_3$ have a conflicting bit common with the block having the element $n_1$ and the block having the element $n_5$ have a conflicting bit common with the block having the element $n_2$. In this case, we send the block having the element $n_4$ to table $T_0$ and the rest of the block lying on the dotted line containing this block to table $T_1$. We send the block having the element $n_1$ to table $T_0$, and the rest of the block lying on this dotted line to table $T_1$. Also, we send the block having the element $n_5$ to table $T_0$, and the rest of the empty block lying on the dotted line containing this block to table $T_1$. Rest all the empty blocks are sent to table $T_0$.

If the blocks having element from $S_2$ and $S_3$ have a conflicting bit common with the same block having the element from $S_1$, then we send the blocks having a conflicting bit common from $S_2$ and $S_3$ to table $T_0$, and the rest of the blocks lying on these dotted lines to table $T_1$. Also, we send the blocks having elements from $S_1$ to table $T_1$. Rest all the empty blocks are sent to table $T_0$.\\ 
\\		
\textbf{Case 6.1.2}
Only one block having the element from $S_2$ or $S_3$ have a conflicting bit common with the block having the element from $S_1$.\\
\\			
\textbf{Case 6.1.2.1} 
Block having an element from $S_2$ have a conflicting bit common with the block having an element from $S_1$. 
Without loss of generality let us say that the block having the element $n_3$ have a conflicting bit common with the block having the element $n_1$. Now, we can have two cases, either the block containing element $n_5$ have a conflicting bit common with the block containing element $n_4$ or it does not. 

If the block containing $n_5$ have a conflicting bit common with the block containing $n_4$, then we send the block containing $n_3$ to table $T_0$, and the rest of the block lying on the dotted line containing this block to table $T_1$. We send the block containing $n_5$ to table $T_0$, and the rest of the block lying on the dotted line containing this block to table $T_1$. Now we see whether the dotted line which contains block having element $n_5$ passes through block having element $n_2$ or not. Let us first consider the case where dotted line which contains block having element $n_5$ passes through block having element $n_2$. In this case, we send the block having the element $n_2$ to table $T_0$, and the rest of the block lying on this dotted line to table $T_1$. Rest all the empty block are sent to table $T_0$. On the other hand, if the dotted line which contains block having element $n_5$ do not pass through block having element $n_2$, then we send the blocks having elements $n_1$ and $n_2$ to table $T_1$, and rest all the empty blocks to table $T_0$.

If the block containing $n_5$ do not have a conflicting bit common with the block having the element $n_4$, then we send the block having the element $n_3$ to table $T_0$ and the rest of the block lying on this dotted line to table $T_1$. Now we see the position of the block which contains the element $n_5$ to make the assignment. If the block which contains element $n_5$ have conflicting bit common with empty block lying on the dotted line which contains block having elements from $S_1$, then we send the block having the element $n_5$ to table $T_1$, and rest all the empty blocks lying on the dotted line which contains this block to table $T_0$. Also, we send blocks having elements from $S_1$ to table $T_1$, and the rest of the empty blocks lying on the dotted line which contains this block to table $T_0$. Rest of the empty blocks we send to table $T_0$. 

Rest for all other positions of $n_5$, we send blocks having elements $n_3$  to table $T_0$, and all the blocks lying on the dotted line containing this block to table $T_1$. Further, we send the block having the element $n_5$ to table $T_0$, and all other blocks lying on the dotted line containing this block to table $T_1$. Now we see whether the dotted line which contains block having element $n_5$ passes through block having element $n_1$ or $n_2$. Without loss of generality, let us say that dotted line which contains block having element $n_5$ passes through block having element $n_2$, in this case we send the block having element $n_2$ to table $T_0$, and rest all the blocks lying on the dotted line containing this block to table $T_1$.  Rest all the empty blocks are sent to table $T_0$. On the other hand if the dotted line which contains block having element $n_5$ do not pass through block having element $n_1$ or $n_2$, then we send the blocks having elements $n_1$ and $n_2$ to table $T_1$, and rest all the empty blocks to table $T_0$.\\  
\\			
\textbf{Case 6.1.2.2}
Block having an element from $S_3$ have a conflicting bit common with the block having an element from $S_1$. Without loss of generality let us say that block having element $n_5$ have a conflicting bit common with the block having the element $n_1$. In this case, we send the block having the element $n_5$ to table $T_0$ and the rest of the block lying on the dotted line which contains this block to table $T_1$. Now we see the position of the block having the element from $S_2$. 

One of the block having an element from $S_2$ have a conflicting bit common with the empty block on the dotted line which contains block having the element $n_5$. Without loss of generality let us say that block having element $n_3$ have a conflicting bit common with an empty block on the dotted line which contains block having the element $n_5$. In this case, we send the block having an element $n_3$ to table $T_0$, and the rest of the block lying on the dotted line containing this block to table $T_1$. Now we see the position of the block having element $n_4$, if it has a conflicting bit common with a empty block lying on the dotted line which contains block having element from $S_1$, then we send the block having element $n_1$ and $n_2$ to table $T_1$, and the rest of the empty block lying on the  dotted line containing these blocks to table $T_0$. Rest all the empty blocks are sent to table $T_0$. On another hand, if the block having the element $n_4$ does not have a conflicting bit common with a block lying on the dotted line which contains block having the element from $S_1$ then we send the block having the element $n_2$ to table $T_0$, and the rest of the block lying on the dotted line containing  this block to table $T_1$. We send the rest of the empty block to table $T_0$.

If the block having the element from $S_2$ do not have a conflicting bit common with a block lying on the dotted line which contains block having the element $n_5$ then we send the blocks having an element from $S_2$ to table $T_1$, and the rest of the empty block lying on the dotted line which contains these blocks to table $T_0$. Also, we send the blocks having the element $n_1$ and $n_2$ to table $T_1$. We send rest of the empty block to table $T_0$.\\ 
\\			
\textbf{Case 6.1.3}
None of the block having elements from $S_2$ or $S_3$ has a conflicting bit common with the blocks having an element from $S_1$ in table $T_1$. Now, here, we can have two cases. Either the block having the element from $S_3$ have a conflicting bit common with the block having the element from $S_2$ or it does not.

Let us first consider the case where one of the blocks having the element from $S_3$ have a conflicting bit common with one of the blocks having an element from $S_2$. Without loss of generality, we can say that block having the element $n_5$  have a conflicting bit common with the block having the element $n_4$. Now, we see whether the dotted lines containing blocks having elements from $S_2$ or $S_3$ passes through block having elements from $S_1$ or not. Without loss of generality, let us say that dotted line which contains block having element $n_5$ passes through block having element $n_1$. In this case, we send the block having the element $n_5$ to table $T_1$, and the rest of the empty blocks lying on the dotted line which contains this block to table $T_0$. We send the block having the element $n_4$ to table $T_0$, and rest of the block lying on the dotted line which contains this block to table $T_1$. Now we see the positions of the blocks having an element from  $S_1$. Let us first consider the case where a block having an element from $S_1$ have a conflicting bit common with a block lying on the dotted line which contains blocks having elements from $S_2$. Without loss of generality let us say block having the element $n_2$ have a conflicting bit common with block lying on the dotted line which contains blocks having elements from $S_2$. In this case, we send the block having the element $n_2$ to table $T_0$, and all other block lying on the dotted line containing this block to table $T_1$. Rest all the empty blocks are sent to table $T_0$. On the other hand,  if the blocks having elements from $S_1$ do not have a conflicting bit common with block lying on the dotted line which contains blocks having elements from $S_2$, then we send the block having elements from $S_1$ to table $T_1$, and rest all the empty blocks to table $T_0$. If none of the dotted lines which contains blocks having elements from $S_2$ or $S_3$ passes through block having elements from $S_1$, then we send the blocks having elements $n_4$ and $n_5$ to table $T_0$, and all the blocks lying on the dotted lines containing these blocks to table $T_1$. Further, we send the blocks having elements from $S_1$ to table $T_1$, and rest all the empty blocks to table $T_0$.

If the block having an element from $S_3$ do not have a conflicting bit common with block having an element from $S_2$, then we send all the block having elements to table $T_1$, and all the empty block to table $T_0$.\\
\\
\textbf{Case 6.2}
Now we consider the case where one of the sets having elements lie on a dotted line and other set having element lie on the different dotted lines. Without loss of generality consider the case where blocks having elements belonging to  $S_1$ lies on a dotted line and blocks having elements belonging to $S_2$ lies on the different dotted lines.\\
\\			
\textbf{Case 6.2.1}
All the blocks having elements from $S_2$ and $S_3$ have a  conflicting bit common with the blocks having elements from $S_1$. In this case, we send the block having an element from $S_2$ and $S_3$ to table $T_0$, and rest of the empty blocks lying on the dotted lines containing these blocks to table $T_1$. Also, we send the blocks having elements from $S_1$ to table $T_1$ and the rest of the empty blocks lying on the dotted line which contains these blocks to table $T_0$. We send the rest of the empty blocks to table $T_0$.\\ 
\\		
\textbf{Case 6.2.2}
Two blocks having elements from $S_2$ and $S_3$ have a  conflicting bit common with the block having an element from $S_1$. Now, here we can have two cases, either those two blocks have elements belonging to $S_2$, or we can have one element belonging to $S_2$, and other to $S_3$.

Let us first consider the case where two blocks having elements from $S_2$ have a  conflicting bit common with blocks having an element from $S_1$. Without loss of generality, let us say that block having element $n_3$ have a  conflicting bit common with the block having the element $n_1$ and the block having the element $n_4$ have a  conflicting bit common with the block having $n_2$. In this case, we send the block having elements $n_3$ and $n_4$ to table $T_0$ and rest of the empty block lying on the dotted lines containing these blocks to table $T_1$. Now we see the position of the block having the element $n_5$. Here we can have two cases either the dotted line which contains the block having the element $n_5$ passes through one of the block having an element from $S_1$ or it does not. Let us first consider the case where it passes through one of the blocks having elements from $S_1$. Without loss of generality let us say that the dotted line which contains block having the element $n_5$ passes through the block having the element $n_1$. In this case, we send the block having the element $n_5$ to table $T_0$, and rest of the block lying on the dotted line which contains this block to table $T_1$. Also, we send the block having the element $n_1$ to table $T_0$, and rest of the block lying on the dotted line containing this block to table $T_1$. Rest all the empty blocks are sent to table $T_0$. On another hand, if the dotted line which contains the block having the element $n_5$ do not pass through blocks having elements from $S_1$, then we send the block having the element $n_5$ to table $T_0$, and rest of the empty blocks on the dotted line containing this block to table $T_1$. Also, we send the blocks having elements from $S_1$ to table $T_1$, and rest of the empty blocks to table $T_0$. If the block having elements $n_3$ and $n_4$ conflicts with the same block having elements from $S_1$, say $n_1$, then assignment made above will work in this case.

Now we consider the case where one of the blocks having an element from $S_2$ and the block having an element from $S_3$ have a  conflicting bit common with the block(blocks) having an element from $S_1$. Without loss of generality say blocks having the element $n_3$ and $n_5$ have a  conflicting bit common with the block having elements from $S_1$. Now here we can have two cases either blocks having elements $n_3$ and $n_5$ have a  conflicting bit common with the same block having an element from $S_1$ or it have a  conflicting bit common with different blocks having elements from $S_1$. Let us first consider the case where blocks having element $n_3$ and $n_5$ have a  conflicting bit common with the same block having an element from $S_1$ say $n_1$. In this case, we send the blocks having elements $n_3, n_4$ and $n_5$ to table $T_0$, and rest of the blocks lying on the dotted lines containing these blocks to table $T_1$. Now we see whether the dotted line which contains block having element $n_4$ passes through block having element from $S_1$ or not. Without loss of generality let us say that dotted line which contains block having element $n_4$ passes through block having element $n_2$. In this case, we send the block having the element $n_2$ to table $T_0$, and rest of the block which lies on the dotted line containing this block to table $T_1$. Rest all the empty blocks are sent to  table $T_0$. If the dotted line which contain block having element $n_4$ do not pass through block having element from $S_1$, then we send the blocks having elements from $S_1$ to table $T_1$, rest all the empty blocks to table $T_0$. Now we consider the case where blocks having elements $n_3$ and $n_5$ have a  conflicting bit common with different blocks having elements from $S_1$. Without loss of generality let us say that block having element $n_3$ have a  conflicting bit common with the block having the element $n_1$ and the block having the element $n_5$ have a  conflicting bit common with the block having the element $n_2$. In this case also assignment made above will work.\\   
\\
\textbf{Case 6.2.3}
Only one block having an element from $S_2$ or $S_3$ have a  conflicting bit common with the block having an element from $S_1$ in table $T_1$. Now here we can have two cases either the block having an element from $S_2$ have a  conflicting bit common or the block having an element from $S_3$ have a  conflicting bit common with the block having an element from $S_1$.\\
\\ 
\textbf{Case 6.2.3.1} The block having an element from $S_2$ have a  conflicting bit common with a block having an element from $S_1$ in table $T_1$. Without loss of generality let us say that the block having the element $n_3$ conflicts  with the block having the element $n_1$. Now we see the position of the blocks having the element $n_4$ and $n_5$. Let us first consider the case where blocks having elements $n_4$ and $n_5$ have a  conflicting bit common.

If the blocks having elements $n_4$ and $n_5$ have a  conflicting bit common on the dotted line which contains blocks having elements from $S_1$, then we send the blocks having elements $n_3,n_4$ and $n_5$ to table $T_0$, and all the blocks lying on the dotted lines containing these blocks to table $T_1$. We send the block having element $n_1$ and $n_2$ to table $T_1$. Rest all the empty blocks are sent to table $T_0$.

If blocks having elements $n_4$ and $n_5$ have a  conflicting bit common outside the dotted line which contains blocks having elements from $S_1$, then we  see the positions of the dotted lines which contains block having elements $n_4$ and $n_5$. Without loss of generality let us say that the dotted line which contains block having element $n_5$ passes through block having element $n_2$, and the dotted line which contains block having element $n_4$ passes through block having element $n_1$. In this case, we send the blocks having elements $n_2,n_3$ and $n_5$ to table $T_0$, and all the blocks lying on the dotted lines containing these blocks to table $T_1$. Further, we send the block having element $n_4$  to table $T_1$. Rest all the empty blocks are sent to table $T_0$. Now we consider the case where only one dotted line which contains block having element $n_4$  or $n_5$ passes through block having element from $S_1$. Without loss of generality let us say that block having element $n_5$ passes through block having element $n_2$. In this case, we send the blocks having elements $n_2, n_3, n_4$ and $n_5$ to table $T_0$, and rest all the blocks lying on the dotted lines containing these block to table $T_1$. Further, we send rest all the empty blocks to table $T_0$. If none of the dotted lines, which contains blocks having element $n_4$ or $n_5$ passes through block having elements from $S_1$, then we send the blocks having elements  $n_3, n_4$ and $n_5$ to table $T_0$, and all the empty blocks lying on the dotted lines containing these blocks to table $T_1$. Further, we send the blocks having elements from $S_1$ to table $T_1$, and rest all the empty blocks are sent to table $T_0$. 

Now we are left with the case where blocks having elements $n_4$ and $n_5$ do not have a  conflicting bit common. This can have several sub-cases. Let us first consider the case where both the blocks having elements $n_4$ and $n_5$ passes through the dotted line which contains blocks having elements from $S_1$. In this case, we send the blocks having elements $n_3,n_4$ and $n_5$ to table $T_0$, and all the blocks lying on the dotted lines containing these blocks to table $T_1$. We send the blocks having elements $n_1$ and $n_2$ to table $T_1$. Rest all the empty blocks are sent to table $T_1$.

Now let us consider the case where $n_4$ and $n_5$ do not have a conflicting bit common with block lying on the dotted line which contains block having elements from $S_1$. In this case, we send the blocks having elements $n_3,n_4$ and $n_5$ to table $T_1$, and all the blocks lying on the dotted lines containing these blocks to table $T_0$. Block having the element $n_1$ is sent to table $T_0$, and rest all the blocks lying on the dotted line containing these blocks to table $T_1$. Rest all the empty blocks are sent to table $T_0$.

Now, we can also have a case where only one block having the element $n_4$ or $n_5$ have a  conflicting bit common with block lying on the dotted line which contains blocks having elements from $S_1$. Let us first consider the case where block having the element $n_4$ have a  conflicting bit common with the block lying on the dotted line which contains block having elements from $S_1$. In this case, we send the blocks having elements $n_3,n_4$ and $n_5$ to table $T_0$, and all other blocks lying on the dotted line containing these blocks to table $T_1$. Now, we see whether the block having elements from $S_1$ have a conflicting bit common with block lying on the dotted line which contains block having the element $n_5$. Without loss of generality let us say block having the element $n_1$ have a conflicting bit common with block lying on the dotted line which contains block having the element $n_5$. In this case, we send the block having the element $n_1$ to table $T_0$, and all other blocks lying on the dotted line which contains this block to table $T_1$. Rest all the empty blocks are sent to table $T_1$. On another hand, if the dotted line which contains block having the element $n_5$ do not pass through block having element from $S_1$, then we send the block having element from $S_1$ to table $T_1$, and all the blocks lying on the dotted line containing this block to table $T_0$. Rest all the empty blocks are sent to table $T_0$. Now let us consider the case where block having the element $n_5$ have a conflicting bit common with block lying on the dotted line which contains block having elements from $S_1$. In this case, we send the blocks having elements $n_3,n_4$ and $n_5$ to table $T_0$ and all other blocks lying on the dotted line containing these blocks to table $T_1$. Now we see whether the dotted line which contains block having element $n_4$ passes through block having element $n_1$ or $n_2$. Without loss of generality let us say that dotted line which contains block having element $n_4$ passes through block having element $n_2$. In this case, we send the block having element $n_2$ to table $T_0$, rest all the blocks lying on the dotted line which contains this block to table $T_1$. Rest all the empty blocks are sent to table $T_0$. If the dotted line which contains block having element $n_4$ do not pass through block having element from $S_1$, then we send the blocks having elements from $S_1$ to table $T_1$, and rest all the empty blocks to table $T_0$.

Now let us consider the case where none of the blocks having elements $n_4$ or $n_5$ conflict with the dotted line which contains block having elements from $S_1$. In this case, we send the blocks having elements $n_3, n_4$ and $n_5$ to table $T_0$, and all the blocks lying on the dotted lines containing these blocks to table $T_1$. Further, blocks having elements from $S_1$ are sent to table $T_1$, and rest all the empty blocks are sent to table $T_0$. \\
\\			
\textbf{Case 6.2.3.2}
The block having an element from $S_3$ have a conflicting bit common with a block having an element from $S_1$ in table $T_1$. Without loss of generality let us say that block having element $n_5$ have a conflicting bit common with the block having the element $n_1$. Now we see the position of the blocks having elements from $S_2$.

Both the blocks having an element from $S_2$ have a conflicting bit common with block lying on the dotted line which contains block having an element from $S_1$ in table $T_1$. In this case, we send the blocks having the elements $n_3,n_4$ and $n_5$ to table $T_0$ and all the blocks lying on the dotted lines containing these blocks to table $T_1$. Blocks having elements from $S_1$ are sent to table $T_1$ and rest all the empty blocks are sent to table $T_0$.

Both the blocks having an element from $S_2$ do not have a conflicting bit common with block lying on the dotted line which contains blocks having an element from $S_1$. Now we see whether the blocks having elements from $S_2$ conflicts with block having element $n_5$. Let us first consider the case where none of the blocks having elements from $S_2$ conflicts with block having element $n_5$. In this case, we send the blocks having elements $n_3,n_4$ and $n_5$ to table $T_1$, and all the empty blocks lying on the dotted lines containing these blocks to table $T_0$. Further, we send the block having the element $n_1$ to table $T_0$, and all the blocks lying on the dotted line containing this block to table $T_1$. Rest all the empty blocks are sent to table $T_0$. Now let us consider the case where only one block having element from $S_2$ conflicts with block having element $n_5$. Without loss of generality let us say that block having element $n_4$ conflicts with block having element $n_5$. In this case, we send the blocks having elements $n_3, n_4$ and $n_5$ to table $T_0$, and all the blocks lying on the dotted lines containing these blocks to table $T_1$.Now we see whether the dotted line which contains block having element $n_3$ passes through block having element $n_1$ or $n_2$. Without loss of generality let us say that the dotted line which contains block having element $n_3$ passes through block having element $n_2$. In this case, we send the block having element $n_2$ to table $T_0$, all the blocks lying on the dotted line containing this block to table $T_1$. Rest all the empty blocks are sent to table $T_0$. If the dotted line which contains block having element $n_3$ do not pass through block having element from $S_1$, then we send the block having element from $S_1$ to table $T_1$, rest all the empty blocks to table $T_0$. Now we consider the case where both the blocks having elements from $S_2$ conflicts with block having element $n_5$. In this case, we send the blocks having elements $n_3, n_4$ and $n_5$ to table $T_0$, all the blocks lying on the dotted lines containing these blocks to table $T_1$. Further, we send the block having element $n_2$ to table $T_0$, all the blocks lying on the dotted line containing this block to table $T_1$. Rest all the empty blocks are sent to table $T_0$.
Only one block having the element from $S_2$ have a conflicting bit common with block lying on the dotted line which contains blocks having elements from $S_1$ in table $T_1$. Without loss of generality let us say that block having element $n_3$ have a conflicting bit common with block lying on the dotted line which contains block having an element from $S_1$. Now we see whether the block having element $n_4$ conflicts with block having element $n_5$. Let us first consider the case where blocks having element $n_4$ and $n_5$ conflicts. In this case, we send the blocks having elements $n_1, n_2, n_3$ and $n_4$ to table $T_1$, and all the empty blocks lying on the dotted lines which contain these blocks to table $T_0$. We send the blocks having the element $n_5$ to table $T_0$, and rest all the blocks lying on the dotted lines which contain these blocks to table $T_1$. Rest all the empty blocks are sent to table $T_0$. If the blocks having elements $n_4$ and $n_5$ do not conflict, then we send the blocks having elements $n_3, n_4$ and $n_5$ to table $T_0$, and all the empty blocks lying on the dotted line containing this block to table $T_1$. Now we see whether a block having element from $S_1$ conflicts with empty block lying on the dotted line which contains block having element $n_4$. Without loss of generality, let us say that block having element $n_2$ conflicts with empty block lying on the dotted line which contains block having element $n_4$. In this case, we send the block having element $n_2$ to table $T_0$, and rest all the blocks lying on the dotted line containing this block to table $T_1$. Rest all the empty blocks are sent to table $T_0$. If none of the block having element from $S_1$ conflicts with empty block lying on the dotted line which contains block having element $n_4$, then we send the blocks having elements from $S_1$ to table $T_1$, and rest all the empty blocks to table $T_0$. \\
\\
\textbf{Case 6.2.4}
None of the blocks having an element from $S_2$ or $S_3$ have a conflicting bit common with the block having elements from $S_1$ in table $T_1$. In this case, we can have either the block having the element $n_5$ have a conflicting bit common with the block having elements from $S_2$ or it does not. 

Let us first consider the case where the block having the element $n_5$ have a conflicting bit common with one of the blocks having the element from $S_2$. Without loss of generality let us say that the block having the element $n_5$ have a conflicting bit common with the block having the element $n_3$. Now we see whether the block having the element $n_4$ have conflicting bit common with a block lying on the dotted line which contains block having elements from $S_1$ or not. Let us first consider the case where block having the element $n_4$ have conflicting bit common with block lying on the dotted line which contains block having an element from $S_1$. In this case, we send the block having the element $n_3$ to table $T_1$, and all the blocks lying on the dotted line containing this block to table $T_0$. We send the block having the element $n_5$ to table $T_0$, and all the blocks lying on the dotted line which contains this block to table $T_1$. Now we see whether any of the block having elements from $S_1$ have a conflicting bit common with block lying on the dotted line which contains element $n_5$ or not. Without loss of generality let us say that the block having the element $n_1$ have conflicting bit common with block lying on the dotted line which contains block having the element $n_5$. In this case, we send the blocks having the elements $n_1$ and $n_4$ to table $T_0$, and all other blocks lying on the dotted lines which contain these blocks to table $T_1$. Rest all the empty blocks are sent to table $T_0$. If none of the blocks having elements from $S_1$ have conflicting bit common with block lying on the dotted line which contains block having the element $n_5$, then we send the blocks having elements $n_1,n_2$ and $n_4$ to table $T_1$, and rest all the empty blocks to table $T_0$. Now let us consider the case where block having the element $n_4$ do not have conflicting bit common with block lying on the dotted line which contains block having elements from $S_1$. In this case, we see whether a block having element $n_5$ have conflicting bit common with a block having element $n_4$. Let us first consider the case where a block having element $n_5$ do not have conflicting bit common with a block having element $n_4$. In this case, we send the blocks having elements $n_4$ and $n_5$ to table $T_1$, and all the blocks lying on the dotted lines containing these blocks to table $T_0$. Further, we send the block having element $n_3$ to table $T_0$, and all the empty blocks lying on the dotted line containing this block to table $T_1$. Now, we see if either of the blocks having element $n_1$ or $n_2$ have conflicting bit common with block lying on the dotted line which contains block having element $n_3$. Without loss of generality, let us say block having element $n_1$ have conflicting bit common with block lying on the dotted line which contains block having element $n_3$. In this case, we send the block having element $n_1$ to table $T_0$, and all the blocks lying on the dotted line containing this block to table $T_1$. Rest all the empty blocks are sent to table $T_0$. On the contrary, if neither of the block having element $n_1$ or $n_2$ has conflicting bit common with block lying on the dotted line which contains block having element $n_3$, then we send the block having an element from $S_1$ to table $T_1$, and rest all the empty to table $T_0$. We are now left with the case where block having element $n_4$ have conflicting bit common with a block having element $n_5$. In this case, we send the block having $n_3$ and $n_4$ to table $T_1$, and all the blocks lying on the dotted line containing this block to table $T_0$. Further, we send the block having element $n_5$ to table $T_0$, and all the empty blocks lying on the dotted line containing this block to table $T_1$. Now we see whether any of the block having elements from $S_1$ have a conflicting bit common with block lying on the dotted line which contains element $n_5$ or not. Without loss of generality let us say that the block having the element $n_1$ have conflicting bit common with block lying on the dotted line which contains block having the element $n_5$. In this case, we send the block having element $n_1$ to table $T_0$, and all the blocks lying on the dotted line containing this block to table $T_1$. Rest all the empty blocks are sent to table $T_0$. If none of the blocks having elements from $S_1$ have conflicting bit common with block lying on the dotted line which contains block having the element $n_5$, then we send the blocks having elements from $S_1$ to table $T_1$ and rest all the empty blocks to table $T_0$.

Now we are left with the case where block having the element $n_5$ do not have a conflicting bit common with the block having an element from $S_2$. In this case, we send the blocks having elements to table $T_1$ and all the empty blocks to table $T_0$.\\
\\		
\textbf{Case 6.3}
All the elements belonging $S_1,S_2$ and $S_3$ lies on the different dotted line. In this case, we send the blocks having elements to table $T_0$, and all the empty blocks to table $T_1$.\\
\\
\textbf{Case 7}
Two blocks having elements belonging to $S_1=\{n_1,n_2\}$ lies in a same superblock and the element $n_3,n_4$ and $n_5$ to the different superblocks. If the blocks having elements from $S_1$ lies on the different dotted line, then assignment made in Case 5.3 can be used. So let us consider the case where blocks having the element from $S_1$ lies on the same dotted line.\\
\\
\textbf{Case 7.1}
Only one block having element from $n_3,n_4$ and $n_5$ have conflicting bit common with the block having element from $S_1$. Without loss of generality let us say block having the element $n_3$ have conflicting bit common with the block having the element $n_1$.\\
\\		
\textbf{Case 7.1.1}
Blocks having elements  $n_4$ and $n_5$ have conflicting bit common.\\ 
\\			
\textbf{Case 7.1.1.1}
Blocks having elements $n_4$ and $n_5$ have conflicting bit common with block lying on the dotted line which contains block having an element from $S_1$. In this case, we send the block having elements $n_3,n_4$ and $n_5$ to table $T_0$, and rest of the empty blocks lying on the dotted lines containing these blocks to table $T_1$. Blocks having the elements $n_1$ and $n_2$ are sent to table $T_1$ and rest all the empty blocks are sent to table $T_0$.\\  
\\
\textbf{Case 7.1.1.2}
Blocks having elements $n_4$ and $n_5$  does not have conflicting bit common with block lying on the dotted line, which contains block having an element from $S_1$. Now here we can have several cases depending upon whether the dotted lines which contain blocks having elements $n_4$ and $n_5$ passes through blocks having elements from $S_1$ or not. 

Without loss of generality let us say that dotted line which contains block having the element $n_4$ passes through the block having the element $n_1$, and the dotted line which contain block having the element $n_5$ passes through the block having the element $n_2$. Now in this case we send the blocks having elements $n_4$ and $n_1$ to table $T_1$. Also, we send the blocks having elements $n_3,n_5$ and $n_2$ to table $T_0$, and rest of the blocks lying on the dotted lines containing these blocks to table $T_1$. We send rest of the empty blocks to table $T_0$.

Now consider a case where only one dotted line which contains a block having element $n_4$ or $n_5$ passes through the block having an element from $S_1$. Without loss of generality let us say that block having the element $n_4$ passes through the block having an element from $S_1$. Let us first consider the case where the dotted line which contains block having the element  $n_4$ passes through the block having the element $n_2$. In this case, we send the blocks having elements $n_2, n_3, n_4$ and $n_5$ to table $T_0$, and rest of the blocks lying on the dotted lines containing these blocks to table $T_1$. Rest all the empty blocks are sent to table $T_0$. Now we can also have a case where block having the element $n_4$ passes through the block having the element $n_1$. In this case, we send the block having the elements $n_1,n_3,n_4$ and $n_5$ to table $T_0$, and rest of the blocks lying on the dotted lines containing these blocks to table $T_1$. Rest all the empty blocks are sent to table $T_0$. If none of the dotted lines which contains blocks having elements $n_4$ and $n_5$ passes through blocks having elements from $S_1$ then we send the blocks having elements $n_3,n_4$ and $n_5$ to table $T_0$, and all the blocks lying on the dotted lines containing these blocks to table $T_1$. Also, we send the blocks having elements from $S_1$ to table $T_1$. Rest all the empty blocks are sent to table $T_0$. \\      
\\
\textbf{Case 7.1.2}
Blocks having the elements $n_4$ and $n_5$ do not have conflicting bit common in table $T_1$ .\\
\\ 
\textbf{Case 7.1.2.1}
Both the blocks having element $n_4$ and $n_5$ have conflicting bit common with block lying on the dotted line which contains block having an element from $S_1$. In this case, we send the blocks having elements $n_3,n_4$ and $n_5$ to table $T_0$, and rest of the empty blocks lying on the dotted line containing these blocks to table $T_1$. Also, we send the blocks having elements $n_1$ and $n_2$ to table $T_1$, and rest of the empty blocks to table $T_0$.\\ 
\\
\textbf{Case 7.1.2.2}
One of the blocks having an element  $n_4$ or $n_5$ have conflicting bit common with block lying on the dotted line which contains blocks having elements from  $S_1$ and other do not have conflicting bit common. Without loss of generality let us say block having the element $n_4$ have conflicting bit common with block lying on the dotted line which contains blocks having elements from  $S_1$, and block having the element $n_5$ lies outside it. In this case, we send the blocks having elements $n_3,n_4$ and $n_5$ to table $T_0$, and all the blocks lying on the dotted lines containing these blocks to  table $T_1$. If any block having an element from $S_1$ have conflicting bit common with block lying on the dotted line which contains block having elements $n_5$, then we send that block to table $T_0$ all other blocks lying on the dotted line containing this block to table $T_1$. Rest all the empty blocks are sent to table $T_0$. On another hand, if none of the blocks having elements from $S_1$ have conflicting bit common with block lying on the dotted line which contains block having elements $n_5$ then we send the block having elements from $S_1$ to table $T_1$, and rest all the empty blocks to table $T_0$.  \\
\\ 
\textbf{Case 7.1.2.3}
None of the blocks having elements $n_4$ or $n_5$ have conflicting bit common with block lying on the dotted line which contains blocks having elements from $S_1$. Now we see position of the block having element $n_4$ and $n_5$. Let us first consider the case where blocks having elements $n_4$ and $n_5$ do not conflicts with block having element $n_3$. In this case, we send the blocks having elements $n_2,n_3,n_4$ and $n_5$ to table $T_1$. Further, we send the block having the element $n_1$ to table $T_0$ and the rest of the empty blocks lying on the dotted line which contains this block to table $T_1$. Rest all the empty blocks are sent to table $T_0$. If both the blocks having elements $n_4$ and $n_5$ conflicts with block having element $n_3$. In this case, we send the blocks having elements $n_3, n_4$ and $n_5$ to table $T_0$, and all the blocks lying on the dotted lines which contain these blocks to table $T_1$. Further, we send the block having element $n_2$ to table $T_0$, and all the blocks lying on the dotted line which contain this block to table $T_1$. Rest all the empty blocks are sent to table $T_0$. Now we are left with the case, where only one block having element $n_4$ or $n_5$ conflicts with block having element $n_3$. Without loss of generality let us say that block having element $n_4$ conflicts with block having element $n_3$. In this case, we send the block having element $n_3$ and $n_5$ to table $T_0$, and rest all the empty blocks lying on the dotted lines containing theses blocks to table $T_1$. We send the block having element $n_4$ to table $T_1$, and rest all the empty blocks lying on the dotted line containing this block to table $T_0$. Now we see whether the dotted line which contains block having element $n_5$ passes through block having element $n_1$ or $n_2$. Without loss of generality let us say that dotted line which contains block having element $n_5$ passes through block having element $n_1$. In this case, we send the block having element $n_1$ to table $T_0$, and rest all the blocks lying on the dotted line containing this block to table $T_1$. Rest all the empty blocks are sent to table $T_0$. On the other hand, if the dotted line which contains block having element $n_5$ do not pass through block having element $n_1$ or $n_2$, then we send the blocks having elements from $S_1$ to table $T_1$, and rest all the empty blocks to table $T_0$. \\
\\
\textbf{Case 7.2}
Two blocks having elements from $n_3,n_4$ or $n_5$ have conflicting bit common with the block having elements from $S_1$. Without loss of generality let us say that blocks having elements $n_3$ and $n_4$ have conflicting bit common with the block having elements from $S_1$. Now, here we can have two cases, either the blocks having elements $n_3$ and $n_4$ have conflicting bit common with the same block having an element from $S_1$ or it has conflicting bit common with the different blocks having elements from $S_1$.

Let us first consider the case where blocks having the element $n_3$ and $n_4$ have conflicting bit common with the different blocks having elements from $S_1$. Without loss of generality, let us say that block having element $n_3$ have conflicting bit common with the block having the element $n_1$ and the block having the element $n_4$ have conflicting bit common with the block having the element $n_2$. Now let us consider the intersection of the dotted line which contains block having the element $n_5$ from the blocks having elements from $S_1$ . Without loss of generality, let us say that the dotted line which contains block having element $n_5$ passes through the block having element $n_1$, in this case, we send the blocks having elements $n_1,n_3,n_4$ and $n_5$ to table $T_0$ and all other blocks lying on the dotted lines which contains these blocks to table $T_1$. Rest all the empty blocks are sent to table $T_1$. On the another hand, if the block having element $n_5$ do not pass through the blocks having elements from $S_1$, then we send the blocks having $n_3,n_4$ and $n_5$ to table $T_0$, and all other blocks lying on the dotted lines which contains these blocks to table $T_1$. Further, we send the blocks having elements from $S_1$ to table $T_1$ and rest all the empty blocks to table $T_0$.

Now we are left with the case where blocks having the element $n_3$ and $n_4$ have conflicting bit common with only one block having an element from $S_1$. Without loss of generality let  us say that blocks having elements $n_3$ and $n_4$ have conflicting bit common with the block having the element $n_1$. In this case, we see the position of the block having the element $n_5$. If the block having the element $n_5$ have conflicting bit common with empty block lying on the dotted line having elements from $S_1$, then we send the block having the element $n_3,n_4$ and $n_5$ to table $T_0$, and all the empty blocks lying on the dotted lines containing these blocks to table $T_1$. We send the blocks having elements from $S_1$ to table $T_1$. Rest all the empty blocks are sent to table $T_0$.
Further, if the block having the element $n_5$ do not have conflicting bit common with empty block lying on the dotted line which contains blocks having element from $S_1$ then we send the block having element $n_3,n_4$ and $n_5$ to table $T_0$, and all the empty blocks lying on the dotted lines containing these blocks to table $T_1$. Now we see whether the dotted line which contains block having the element $n_5$ passes through a block having an element from $S_1$ or not. Let us first consider the case where the dotted line which contains block having the element $n_5$ passes through a block having an element from $S_1$. Without loss of generality, let us say that the dotted line which contains block having the element $n_5$ passes through the block having the element $n_1$. In this case, we send the block having the element $n_1$ to table $T_0$, and all other blocks lying on the dotted line containing this block to table $T_1$. Rest all the empty blocks are sent to table $T_0$. On another hand, if the dotted line which contains block having the element $n_5$ do not pass through the block having an element from $S_1$, then we send the blocks having the elements from $S_1$ to table $T_1$. Rest all the empty blocks are sent to table $T_0$.\\        
\\
\textbf{Case 7.3} All the three blocks having elements $n_3,n_4$ and $n_5$ have conflicting bit common with the blocks having elements from $S_1$. In this case, we send the blocks having elements $n_3,n_4$ and $n_5$ to table $T_0$, and all the blocks lying on the dotted lines containing these blocks to table $T_1$. Blocks having elements from $S_1$ are sent to table $T_1$. Rest all the empty blocks are sent to table $T_0$.\\
\\ 
\textbf{Case 7.4}
None of the blocks having elements from $n_3,n_4$ and $n_5$ have conflicting bit common with the blocks having elements from $S_1$.

Let us first consider the case where all the blocks having elements $n_3,n_4$ and $n_5$ have conflicting bit common. If all of them have conflicting bit common on the dotted line which contains blocks having elements from $S_1$ then we send the blocks having elements $n_3,n_4$ and $n_5$ to table $T_0$, and rest of the empty blocks lying on the dotted lines containing these blocks to table $T_1$. Also we send the blocks having elements $n_1$ and $n_2$ to table $T_1$. Rest all the empty blocks are sent to table $T_0$.

If the blocks having elements $n_3,n_4$ and $n_5$ have conflicting bit common outside the dotted line which contains blocks having elements from $S_1$, then we see the intersection of the dotted lines containing these blocks with the blocks having elements from $S_1$. Now since all the blocks having elements from $n_3,n_4$ and $n_5$ have conflicting bit common at most two dotted lines having blocks containing these elements can have conflicting bit common with the blocks having elements from $S_1$. Without loss of generality let us say that the dotted line which contains block having the element $n_3$ passes through the block having the element $n_1$ and the dotted line which contains block having the element $n_4$ passes through the block having the element $n_2$. In this case, we send the blocks having element $n_1$ and $n_3$ to table $T_1$. Also, we send the blocks having elements $n_2,n_4$ and $n_5$ to table $T_0$, and rest of the empty blocks lying on the dotted line containing these blocks to table $T_1$. Rest all the empty blocks are sent to table $T_0$.

Without loss of generality let us now consider a case where only one dotted line having element say $n_3$ passes through the block having element say $n_1$, then we send the blocks having elements $n_1,n_2$ and $n_3$ to table $T_1$. Also, we send the blocks having elements $n_4$ and $n_5$ to table $T_0$, and rest of the empty blocks lying on the dotted lines containing these blocks to table $T_1$.Rest all the empty blocks are sent to table $T_0$.

If none of the dotted lines containing blocks having elements $n_3,n_4$ and $n_5$ passes through the blocks having element from $S_1$ then we send the blocks having elements $n_3,n_4$ and $n_5$ to table $T_0$, and rest of the empty blocks lying on the dotted line containing these blocks to table $T_1$. Further, we send the blocks having element $n_1$ and $n_2$ to table $T_1$, and rest of the empty blocks to table $T_0$.

Now let us consider the case where only two blocks having elements from $n_3,n_4$ and $n_5$ have conflicting bit common. Without loss of generality let us say the blocks having elements $n_3$ and $n_4$ have conflicting bit common. Now we see the intersection of dotted lines containing the blocks $n_3$ and $n_4$. Let us first consider the case where both the dotted lines containing blocks $n_3$ and $n_4$ passes through the blocks having elements from $S_1$. Without loss of generality let us say the dotted line which contains block having the element $n_3$ passes through the block having the element $n_1$ and the dotted line which contains block having the element $n_4$ passes through the block having the element $n_2$. Now, we see the position of the block having element $n_5$. Let us first consider the case where block having $n_5$ have conflicting bit common with the block lying on the dotted line which contains block having an element from $S_1$. Further, let us consider that the dotted line which contains block having element $n_5$ intersects the block having element $n_3$ and $n_4$. Now, it is interesting to note that if two nonempty and an empty blocks intersect, then we can either send both nonempty blocks to table $T_1$, or we can send a nonempty block and an empty to  table $T_1$. Let us first consider the case where we can send both the blocks having element $n_3$ and $n_4$ to table $T_1$. In this case, we send all the blocks having elements to table $T_1$, and all the empty blocks to  table $T_0$. Now, without loss of generality let us consider that we can send the blocks having element $n_3$ and the empty block lying on the dotted line which contains block having element $n_5$ to table $T_1$. In this case, we send the block having element $n_2,n_4$ and $n_5$ to table $T_0$, and all the blocks lying on the dotted line containing this block to table $T_1$. Further, we send the block having element $n_3$ to table $T_1$. Rest all the empty blocks are sent to table $T_0$. Let us now consider the case where dotted line which contains block having element $n_5$ do not pass through both the block having element $n_3$ and $n_4$. Now since dotted line which contains block having element $n_5$ do not intersect with both the blocks having elements $n_4$ and $n_5$ it can conflict with at most one block having element. Without loss of generality, let us say that empty block lying on the dotted line which contains block having element $n_5$ either conflict with a block having element $n_3$ or it does not. In this case, we send the blocks having elements $n_1,n_3$ and $n_5$ to table $T_0$, and all the blocks lying on the dotted lines containing these blocks to table $T_1$. Further, we send the blocks having element $n_2$ and $n_4$ to table $T_1$. Rest all the empty blocks are sent to table $T_0$. Let us now consider the case where block having element $n_5$ do not intersect with the dotted line which contains block having elements from $S_1$. Now let us consider the case where block having element $n_5$ do not conflict with block lying on the dotted line having elements from $S_1$. In this case, we send the block having element $n_5$ to table $T_1$ and all the empty blocks lying on the dotted line containing this block to table $T_0$. Now we see if the block having element $n_5$ have conflicting bit common either with block lying on the dotted line which contains block having element $n_3$ or $n_4$. Without loss of generality, let us say block having element $n_5$ have conflicting bit common with empty block lying on the dotted line which contains block having element $n_3$. In this case, we send the block having element $n_1$ and $n_3$ to table $T_1$. Further, we send the block having element $n_2$ and $n_4$ to table $T_0$, and all the blocks lying on the dotted lines containing these blocks to table $T_1$. Rest all the empty blocks are sent to table $T_0$. If the block having element $n_5$, do not have conflicting bit common with block lying on the dotted line which contains block having element $n_3$ or $n_4$ then previous assignment works. 

Let us consider the case where only one dotted line which has block having the element $n_3$ or $n_4$ passes through the block having elements from $S_1$. Without loss of generality let us say that dotted line which contains block having the element $n_3$ passes through the block having the element $n_1$. In this case, we see the position of the block having the element $n_5$. Let us first consider the case where block having the element $n_5$ have conflicting bit common with block lying on the dotted line which contains blocks having elements from $S_1$. In this case,  we send the block having element $n_1, n_3, n_4$ and $n_5$ to table $T_0$, and all the blocks lying on the dotted lines containing these blocks to table $T_1$. Further, we send rest all the empty blocks to table $T_0$. 

Now let us consider the case where the block having the element $n_5$ do not have conflicting bit common with block lying on the dotted line which contains blocks having elements from $S_1$. Now we see whether the block having element $n_5$ conflicts with the dotted line which contain block having element $n_3$ or $n_4$. If the block having element $n_5$ conflicts with block lying on the dotted line which contains block having element $n_3$, then we send the blocks having elements $n_2, n_4$ and $n_5$ to table $T_0$, rest all the blocks lying on the dotted lines containing these blocks to table $T_1$. Further, we send the block having element $n_3$ to table $T_1$, and rest all the empty blocks to table $T_0$. If the block having element $n_5$ conflicts with block lying on the dotted line which contains block having element $n_4$, then we send the blocks having elements $n_4$ and $n_5$ to table $T_1$, and rest all the blocks lying on the dotted lines containing these blocks to table $T_1$. Further, we send the blocks having elements $n_1$ and $n_3$ to table $T_0$, and all the blocks lying on the dotted line containing these blocks to table $T_1$. Rest all the empty blocks are sent to table $T_0$. If the block having element $n_5$ do not conflict with block lying on the dotted line which contains block having element $n_3$ or $n_4$, then we send the blocks having elements $n_1, n_2, n_3$ and $n_5$ to table $T_1$, and all the blocks lying on the dotted lines containing these blocks to table $T_0$. Further, we send the block having element $n_4$ to table $T_0$, and all the empty blocks lying on the dotted line containing this block to table $T_1$. Rest all the empty blocks are sent to table $T_0$.

If none of the dotted lines which contains block having elements $n_3$ or $n_4$ passes through block having element from $S_1$, then we send the blocks having elements $n_3,n_4$ and $n_5$ to table $T_0$, and all the empty blocks lying on the dotted lines containing these blocks to table $T_1$. Now either the block having an element from $S_1$ have conflicting bit common with block lying on the dotted line which contains block having the element $n_5$ or it does not. Let us first consider the case where block having elements from $S_1$ have conflicting bit common with block lying on the dotted line which contains block having the element $n_5$. Without loss of generality let us say that block having element $n_1$ have conflicting bit common with block lying on the dotted line which contains block having the element $n_5$. In this case, we send the block having the element $n_1$ to table $T_0$, and all the blocks lying on the dotted line containing this block to table $T_1$. Rest all the empty blocks are sent to table $T_0$. If none of the block having an element from $S_1$ has conflicting bit common with block lying on the dotted line which contains block having the element $n_5$, then we send the blocks having an element from $S_1$ to table $T_1$. Rest all the empty blocks are sent to table $T_0$. 

Now we are left with the case where none of the blocks having elements  $n_3,n_4$ and $n_5$ have conflicting bit common. In this case, we send the blocks having elements to table $T_1$, and all the empty blocks to table $T_0$.

\end{document}